\documentclass[10pt, conference, letterpaper]{IEEEtran}

\usepackage{framed}
\usepackage{hyperref}
\usepackage{tcolorbox}
\usepackage{cite}
\usepackage{amsmath}
\usepackage{amsfonts}
\usepackage{mathtools}
\usepackage{amssymb}
\usepackage{amsthm}
\usepackage[ruled,linesnumbered,noend]{algorithm2e}
\usepackage{cases}
\usepackage{graphicx}
\usepackage{bbm}
\usepackage{subcaption}

\newtheorem{theorem}{{\bf Theorem}}

\newtheorem{lemma}{{\bf Lemma}}

\newtheorem{corollary}{{\bf Corollary}}
\newtheorem{observation}{{\bf Observation}}

\newtheorem{defn}{{\bf Definition}}

\IEEEoverridecommandlockouts

\begin{document}
%
\title{Scheduling in Multi-Hop Wireless Networks With Deadlines}
\author{Nicholas Jones and Eytan Modiano
\thanks{This material is based upon work supported by the Department of the Air Force under Air Force Contract No. FA8702-15-D-0001. Any opinions, findings, conclusions or recommendations expressed in this material are those of the author(s) and do not necessarily reflect the views of the Department of the Air Force.}
\thanks{The authors are with the Laboratory for Information and Decision Systems (LIDS), Massachusetts Institute of Technology, Cambridge, MA 02139, USA (email: jonesn@mit.edu, modiano@mit.edu)}}

\IEEEaftertitletext{\vspace{-0.6\baselineskip}}
\maketitle

\begin{abstract}
    We analyze the problem of scheduling in wireless networks to meet end-to-end service guarantees, defined by instantaneous throughput and hard packet deadlines. Using a network slicing model to decouple the queueing dynamics between flows, we show that the network’s ability to meet hard deadline guarantees under interference is largely influenced by the link scheduling policy. We characterize throughput- and deadline-optimal policies for a solitary flow operating in isolation, which provide bounds on feasibility in the general case with multiple flows. We prove that packet delays can grow arbitrarily large in the multi-flow setting under a worst-case stabilizing policy, showing that queue stability is not sufficient to guarantee tight deadlines. We derive conditions on end-to-end packet delays in terms of link inter-scheduling times, and show that it is possible to make hard guarantees under any interference model by solving a generalized version of the pinwheel scheduling problem. Finally, we introduce a decentralized polynomial-time algorithm which can meet tight end-to-end packet deadlines while achieving near-optimal throughput.
\end{abstract}

\maketitle

\section{Introduction}~\label{sec:intro}

Next-generation wireless networks require strict throughput and delay guarantees for technologies such as real-time control and inference taking place over the network. Due to the need for high reliability, best-effort service is insufficient to meet these needs. The 5G standard supports Quality of Service (QoS) guarantees at the flow level using network slicing~\cite{tmobile2023,5ghub}, and the level of network traffic with QoS requirements is only expected to grow with 6G and beyond. A largely open question, however, is how to efficiently schedule resource-limited wireless networks to guarantee these QoS requirements are met.

Much of the wireless scheduling literature focuses on online scheduling policies and stochastic arrivals, most notably the well-known Max-Weight policy~\cite{tassiulas1990stability} and its variants. Recently, work has been done on online scheduling to meet hard packet deadlines, with policies that guarantee performance within a constant factor of an optimal offline policy under the same traffic arrivals~\cite{yan2021asymptotically, tsanikidis_scheduling_2024}. Unfortunately, without putting constraints on network traffic, these optimality factors can become large, and it is difficult to make guarantees on the service that any single packet will experience. Motivated by this, we consider traffic flows with deterministic arrivals, and a class of policies which may be computed online but are able to make guarantees on both throughput and end-to-end delay a priori.

We consider wireless networks with general topologies, and develop efficient scheduling policies to meet a given set of QoS requirements. In particular, we study the impact that wireless interference plays on end-to-end packet delay and optimize our scheduling policies to account for this. We use a network slicing model akin to virtual circuits to decouple the queueing dynamics between traffic flows. This model is not only practical for making service guarantees, but useful for highlighting the impact of interference. 

In the offline wireless scheduling literature, there is a large body of work on scheduling for throughput maximization~\cite{jain2003impact,hajek_link_1988,kodialam2003characterizing}. Most notably, Hajek and Sasaki designed an algorithm to find a minimum schedule length which meets a set of link demands in polynomial time~\cite{hajek_link_1988}, and Kodialam et al. used Shannon's algorithm for coloring a multigraph to efficiently schedule links while achieving at least $2/3$ of maximum capacity~\cite{kodialam2003characterizing}.

There has also been considerable work on QoS guarantees in networks. One of the first approaches was Cruz's network calculus~\cite{cruz1991calculus,cruz1991calculus2}, which bounds the delay that each packet experiences over multiple hops in a wired network, when arrivals are constrained by a traffic-shaping envelope. Several works have extended this framework to the wireless setting using a variant called stochastic network calculus~\cite{fidler2006end, burchard2006min, li2007network, al2013min}, which bounds the tails of arrival and service processes to obtain a high-probability bound on end-to-end delay. While network calculus is perhaps the most comprehensive analytical QoS framework in the literature, it fails to capture the interference present in wireless networks. It is certainly possible to fix a scheduling policy respecting interference constraints and use calculus to find delay bounds, but optimizing this schedule quickly becomes intractable.

A novel QoS framework for single-hop wireless networks was developed by Hou and Kumar~\cite{hou2009theory}. It assumes a set of sources, each with a single packet that must be delivered by the end of a frame with some probability, thereby meeting a deadline equal to the frame length. They design a policy to ensure the ``delivery ratio,'' or time average fraction of packets which are delivered within each frame, meets a reliability requirement. They extend this framework in~\cite{hou2010utility} to solve utility maximization and in~\cite{hou2013scheduling} to support Markov arrival processes.

Several works have extended a version of this framework to multi-hop. In~\cite{li2012scheduling}, the authors analyze a multi-hop network with end-to-end deadline constraints and develop policies to meet delivery ratio requirements over wired links. In~\cite{liu2019spatial}, the authors design a spatio-temporal architecture with virtual links to solve a similar problem. In~\cite{singh2018throughput} and~\cite{singh2021adaptive}, the authors analyze a multi-hop wireless network with unreliable links. By considering each packet individually, they develop both centralized and decentralized policies for maximizing throughput with hard deadlines, using relaxed link capacity constraints and assuming no interference.

The works closest to ours are~\cite{djukic2007quality} and~\cite{djukic_delay_2009}, which consider a multi-hop wireless network and find a transmission schedule to meet deadline guarantees under constant traffic arrivals. Similar to the works above, they split time into frames and develop a mixed-integer program to find an optimal ordering of link activations within each frame, bounding packet delay as a linear function of the frame length. They show that the optimal ordering can be found in polynomial time under tree topologies. The authors of~\cite{cappanera_link_2009} and~\cite{cappanera_efficient_2011} generalize the arrival processes to calculus-style envelopes and consider the case of sink-tree networks, while routing is incorporated into~\cite{cappanera_optimal_2013}. The authors of~\cite{chilukuri_delay-aware_2015} show that shortening frame lengths and optimizing slot re-use for non-interfering links can improve efficiency.

In this paper we adopt a novel approach to meeting QoS guarantees under wireless interference by removing the dependency between packet deadlines and the traditional notion of frame length. We show that deadline guarantees can be made by bounding link inter-scheduling times, effectively creating unique frames for each link with endpoints at scheduling opportunities, and interleaved in time with the frames of other links. This greatly enhances the flexibility of the schedule, allowing us to simultaneously achieve tighter deadlines and higher throughput, while ensuring that interference constraints are met. Our main contributions can be summarized as follows.

\begin{itemize}
    \item In Section~\ref{sec:isolatedflows}, we develop conditions on throughput and deadline feasibility for a solitary flow. We introduce an \textit{Ordered Round-Robin} scheduling policy in this setting, provide tight delay bounds, and show that it is deadline-optimal when throughput is below a certain threshold. Above this threshold, we show that a greedy policy is deadline-optimal under a total interference model, when only one link can be scheduled at a time, and is a good heuristic in the general case.
    \item In Section~\ref{sec:generalfeasibility}, we prove a tight upper bound on end-to-end packet delays under the worst-case stabilizing policy, and show that delay can grow arbitrarily large even though queues are stable. We then show that under some mild assumptions, end-to-end delay can be bounded as a function of link inter-scheduling times. Using these conditions, we examine the general problem of multiple flows in an arbitrary topology, and we formulate a feasibility problem to jointly satisfy deadline, rate, and interference constraints. We show that under total interference, the problem reduces to an NP-complete problem known as pinwheel scheduling~\cite{holte1989pinwheel}, and we introduce a generalized version of this problem which we call \textit{pinwheel coloring} for an arbitrary interference model.
    \item In Section~\ref{sec:algorithmdesign}, we design a polynomial-time algorithm called \textit{Weighted Greedy Coloring} ($WGC$), which approximately solves the pinwheel coloring problem when coupled with existing pinwheel schedulers. This algorithm can be run in a decentralized fashion at each node, requiring no control overhead. When it returns a feasible solution, it is guaranteed to be a coordinated, conflict-free schedule that meets rate and deadline constraints for all flows.
    \item Finally, in Section~\ref{sec:simulations}, we demonstrate the performance of our algorithms through extensive simulations, verifying that greedy is a good heuristic for solitary flows and that the $WGC$ algorithm can simultaneously achieve tight deadlines and near-optimal throughput for a variety of network topologies and traffic arrivals.
\end{itemize}

An earlier version of this work appeared in the conference proceedings of ACM MobiHoc 2024~\cite{jones2024optimal}.

\section{Preliminaries}~\label{sec:sysmodel}
\vspace{-15pt}
\subsection{System Model}

We consider a wireless network with fixed topology modeled as a directed graph $G=(V,E)$. Time is slotted, with the duration of a slot dictated by the physical and link layers of the network. Each link $e \in E$ has a fixed capacity $c_e$ and can transmit up to this number of packets in each time slot it is activated. Because links share a wireless channel, they are subject to interference, which restricts the links that can be scheduled at the same time. We assume a centralized controller that computes a pre-determined offline schedule of non-interfering links to activate in each slot, and communicates this schedule to each node ahead of time.

We consider general interference models, which can be described using a conflict graph. Define the conflict graph of the network graph $G$ as $G_c = (V_c, E_c)$, and let each vertex $v \in V_c$ correspond to a link in the network graph $G$. An edge exists between two conflict nodes $v$ and $v'$ if and only if their corresponding links interfere in the network graph. We assume that interference is mutual between links, so $G_c$ is an undirected graph. Any clique in the conflict graph represents a set of mutually interfering links, and any independent set in the conflict graph is a set of links that can be activated simultaneously. We denote the set of all feasible activation sets as $\mathcal{M}$.

In some cases it is useful to consider a subset of interference models $\Phi$, which take the form of $\phi$-hop interference constraints. Specifically, for any $\phi$, no two links can be activated at the same time if they are separated by fewer than $\phi$ hops. We will also refer to the interference model itself as $\phi \in \Phi$. In the special cases of no interference and total interference, where only one link can be activated at a time, we define $\phi$ to be $0$ and $|E|-1$ respectively. We denote $\mathcal{M}^{\phi}$ as the set of feasible activation sets under the model $\phi$.

Traffic arrives at the network in the form of flows, which request a level of service from the network. Each flow $f_i \in \mathcal{F}$ is assigned a fixed pre-determined route $\mathcal{T}^{(i)}$ from source to destination, and we denote $\mathcal{T}^{(i)}_j$ as the $j$-th hop in the route. Arrivals are deterministic, with $\lambda_i$ packets belonging to flow $f_i$ arriving at its source node in each slot\footnote{This can be generalized to a network calculus-style envelope, and for simplicity of exposition we assume traffic shaping occurs before packets arrive at the source.}. We assume a fluid traffic model, so $\lambda_i$ need not be integer, and we colloquially use the terms packet and traffic interchangeably throughout the paper. Flows are further characterized by a deadline $\tau_i$ for each $f_i$, and we say that a packet meets its deadline if it is delivered to its destination (over possibly many hops) within $\tau_i$ slots of when it arrives, otherwise the packet expires. Assume that packet arrivals stop at time $T$, but packets remaining in the network are given time to be served by their respective deadlines. 

Each link $e \in \mathcal{T}^{(i)}$ reserves capacity for flow $f_i$ in the form of a network slice, akin to a virtual circuit, with slice width (i.e., capacity) equal to $w_{i,e}$. We will also refer to the slice itself as $w_{i,e}$ when there is no risk of confusion. Because link capacities are fixed, the sum of all slice widths allocated on link $e$ must be bounded by $c_e$. Each slice has its own first-come-first-served queue, which decouples the queueing dynamics between flows. Let $Q_{i,e}(t)$ be the size of the queue belonging to slice $w_{i,e}$ at the beginning of time slot $t$. We assume the network is empty before $t=0$, so $Q_{i,e}(t) = 0$ for all $i$ and $e$, and all $t < 0$.

\subsection{Policy Structure}\label{sec:policy}

Define $\Pi$ as the set of admissible scheduling policies, which are work-conserving and satisfy interference constraints. Let $\mu^{\pi}(t) \in \mathcal{M}$ be the set of links activated at time $t$ under $\pi$, and let $\mu_e^{\pi}(t) = 1$ if $e \in \mu^{\pi}(t)$ and $0$ otherwise. The work-conserving property ensures that each scheduled link serves the smaller of its queue size and its slice capacity each time it is activated.

We are interested in policies which meet the following performance objective.
\begin{framed}
\begin{defn}\label{def:supportingpolicy}
    A policy $\pi \in \Pi$ supports the set of flows $\mathcal{F}$ if and only if all packets in $\mathcal{F}$ reach their destination without expiring, over any finite time horizon $T$.
\end{defn}
\end{framed}

We first show that to meet this objective we need only consider cyclic scheduling policies. Define a scheduling policy $\pi$ as cyclic if $\mu^{\pi}(t) = \mu^{\pi}(t+K^{\pi})$, for all $t \geq 0$ and some positive integer $K^{\pi}$, which we refer to as the scheduling period. Denote the set of all admissible cyclic scheduling policies as $\Pi_c \subseteq \Pi$.

\begin{framed}
\begin{theorem}\label{th:cyclicschedules}
    If there exists a policy in $\Pi$ that supports $\mathcal{F}$, then there must exist at least one such cyclic policy $\pi \in \Pi_c$, and some $t_0 > 0$, such that
    \begin{equation}
        Q_{i,e}^{\pi}(t) = Q_{i,e}^{\pi}(t+K^{\pi}), \ \forall \ f_i \in \mathcal{F}, \ e \in E, \ t_0 \leq t \leq T.
    \end{equation}
\end{theorem}
\end{framed}

\begin{proof}
    See Appendix~\ref{app:cyclicschproof}.
\end{proof}

Here $t_0$ represents the time required for the network to ``ramp up'' from its initial state. The theorem shows that after this ramp-up period, the network will reach a steady state under any cyclic policy $\pi$, and that we can restrict ourselves to such cyclic policies without loss of optimality. For ease of analysis, we assume $t_0 + K^{\pi} \leq T$, so that the network operates in steady-state for at least one full scheduling period.

Under a cyclic policy $\pi$, define the time average activation rate of link $e$ as
\begin{equation}
    \bar{\mu}_e^{\pi} \triangleq \frac{1}{K^{\pi}} \sum_{t=0}^{K^{\pi}-1} \mu^{\pi}_e(t),
\end{equation}
and the number of activations per scheduling period as $\eta_e^{\pi} \triangleq \bar{\mu}_e^{\pi} K^{\pi}$. Clearly for packets to meet any finite deadline, queues must be bounded by a finite constant. We refer to any queue which remains bounded as $T \to \infty$ as stable, which leads to the following condition on activation rates.
\begin{framed}
\begin{corollary}\label{cor:ratestability}
    Queues are stable under a policy $\pi \in \Pi_c$ if and only if
    \begin{equation}\label{eq:ratestabcond}
        \bar{\mu}_e^{\pi} \geq \frac{\lambda_i}{w_{i,e}}, \ \forall \ e \in \mathcal{T}^{(i)}, \ f_i \in \mathcal{F}.
    \end{equation}
\end{corollary}
\end{framed}

\begin{proof}
    See Appendix~\ref{app:stabilityproof}.
\end{proof}

\begin{figure}
    \centering
    \includegraphics[width=0.45\textwidth]{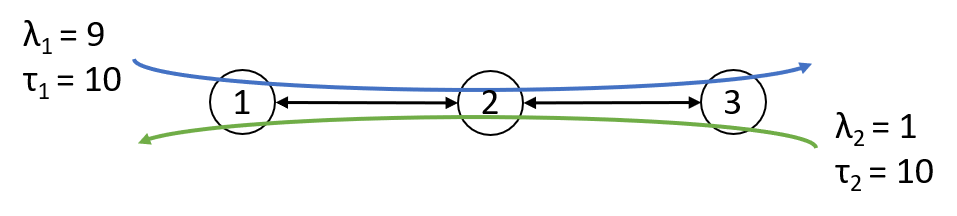}
    \caption{Illustrative Two-Hop Example}
    \label{fig:two-hop-example}
\end{figure}

While stability is necessary for deadlines to be met, it is not sufficient, because queues may still grow too large to satisfy deadlines. Consider the two-hop example in Figure~\ref{fig:two-hop-example}, with flow $f_1$ traveling from node $1$ to node $3$, and $f_2$ traveling in the opposite direction. Both flows have a deadline of $10$ slots, with $\lambda_1 = 9$ and $\lambda_2 = 1$. Assume only one link can be active at a time due to interference, and let the capacity of each link be $27$. Note that while links are drawn as bi-directional, this represents two separate directional links in our model. Because only one flow traverses each link, we can set slice widths equal to link capacity in all cases.

From~\eqref{eq:ratestabcond}, we observe that a necessary condition for supporting $\mathcal{F}$ is that links $(1,2)$ and $(2,3)$ are activated at least $\frac{\lambda_1}{27} = \frac{1}{3}$ of the time, and links $(3,2)$ and $(2,1)$ are activated at least $\frac{\lambda_2}{27} = \frac{1}{27}$ of the time. Consider the cyclic policy
\begin{equation}
    \pi_1 = \{(1,2),(2,3),(3,2),(1,2),(2,3),(2,1)\},
\end{equation}
which repeats itself every $6$ slots, and where $\bar{\mu}_{(1,2)} = \bar{\mu}_{(2,3)} = \frac{1}{3}$ and $\bar{\mu}_{(2,1)} = \bar{\mu}_{(3,2)} = \frac{1}{6}$. This satisfies the necessary activation rates, so all queues are stable under $\pi_1$. In fact, a close examination shows that no more than $27$ packets can arrive at any link in between slots when that link is scheduled, so queues are not only stable but are emptied each time they are served. This means that packets see no queueing delay, and experience only \textit{scheduling delay}, i.e., waiting for a link to be scheduled.

Given this, one can verify that packets arriving at $t=1,4,7,\dots$ experience the largest delay of any $f_1$ packets. These packets must wait $3$ slots before reaching node $2$ and $4$ slots before reaching their destination. Similarly, packets arriving at $t=3,9,15,\dots$ experience the largest delay of any $f_2$ packets. These packets must wait $6$ slots before reaching node $2$, and then another $3$ slots before reaching their destination, resulting in a total delay of $9$ slots. The worst-case delay for both flows satisfies deadline constraints, so $\pi_1$ is a supporting policy. Note that if deadlines were tightened to $\tau_1 = \tau_2 = 8$, the above policy would not satisfy deadlines, but would still remain stable.

Furthermore, there exist many alternative schedules where queues are stable but the original deadlines $\tau_1 = \tau_2 = 10$ are not satisfied. Consider the schedule
\begin{equation}
    \pi_2 = \{(2,3),(2,3),(1,2),(1,2),(2,1),(3,2)\},
\end{equation}
which again repeats itself every $6$ slots, and note that this is simply a re-ordering of $\pi_1$. Because activation rates are unchanged, the queues are guaranteed to be stable as under $\pi_1$, but the re-ordering of the schedule drastically increases the scheduling delay. Consider an $f_1$ packet that arrives at $t=4$. It must wait $5$ slots before it reaches node $2$, followed by another $4$ slots before it reaches node $3$. Similarly, consider an $f_2$ packet that arrives at $t=0$. It must wait $6$ slots before it reaches node $2$, and then another $5$ slots before it reaches node $1$. The worst-case delays increase from $4$ and $9$ slots respectively under $\pi_1$ to $9$ and $11$ slots respectively under $\pi_2$. Simply shuffling the order of the schedule more than doubles packet delays for $f_1$ and causes $f_2$ packets to miss their deadlines. This leads to the following critical observation.
\begin{framed}
    \begin{observation}
    In addition to queue stability, conditions on schedule order are necessary to satisfy deadline guarantees.
\end{observation}
\end{framed}
A large focus of this work is defining these conditions in a meaningful way so that we can design policies which satisfy them in polynomial time. To begin, we extract the following necessary and sufficient condition for meeting deadlines from Theorem~\ref{th:cyclicschedules}.

\begin{framed}
\begin{corollary}\label{cor:deadlinecondition}
    All flow $f_i$ packets meet their deadlines under a scheduling policy $\pi \in \Pi_c$ if and only if
    \begin{equation}\label{eq:queuebound}
        \max_{0 \leq t \leq T} \sum_{e \in \mathcal{T}^{(i)}} Q_{i,e}^{\pi}(t) \leq \lambda_i \tau_i,
    \end{equation}
    for any finite $T$.
\end{corollary}
\end{framed}

\begin{proof}
    Follows from the proof of Theorem~\ref{th:cyclicschedules}.
\end{proof}

Before analyzing the general case, it is instructive to characterize the feasible guarantees that can be made to a solitary flow $f_i$ under a policy $\pi \in \Pi_c$, which is the focus of the next section.

\section{Solitary Flow Feasibility}\label{sec:isolatedflows}

In this section we consider a solitary flow that operates in isolation, so the policy can be optimized independently of other traffic. We can then view a flow $f_i$'s route $\mathcal{T}^{(i)}$ as a line network with the source at one end and the destination at the other. In addition, we assume this line network follows a $\phi$-hop interference model for ease of analysis. Denote the set of jointly achievable rate and deadline guarantees for $f_i$ under an interference model $\phi$ and fixed set of slices $\boldsymbol{w}$ as $\Lambda_i^{\phi}(\boldsymbol{w})$.

In the general setting, with flows on separate routes coupled through the scheduling policy, the jointly achievable region over all flows is a subset of what can be achieved by each flow in isolation. Analyzing $\Lambda^{\phi}_i(\boldsymbol{w})$ allows us to define not only optimal policies for a single flow, but also bounds on feasibility in the general case. Because we are focused on a single flow in the remainder of this section, we drop the subscript $i$ for ease of notation and note the results hold for all flows without loss of generality. We begin by characterizing throughput optimality.

\subsection{Throughput Optimality}

Define $\lambda^*(\pi,\boldsymbol{w})$ as the maximum rate a flow can take as $\tau \to \infty$, under fixed slice widths $\boldsymbol{w}$ and while supported by a policy $\pi$. We omit the dependency on $\phi$ for ease of notation. From Corollary~\ref{cor:ratestability}, it immediately follows that
\begin{equation}\label{eq:lambdabound}
    \lambda^*(\pi,\boldsymbol{w}) = \min_{e \in \mathcal{T}} \bar{\mu}_e^{\pi} w_e,
\end{equation}
and we define a throughput-optimal policy as follows.

\begin{framed}
\begin{defn}
    A policy $\pi \in \Pi_c$ is throughput-optimal for a set of slice widths $\boldsymbol{w}$ if $\lambda^*(\pi,\boldsymbol{w}) \geq \lambda^*(\pi',\boldsymbol{w})$ for all $\pi' \in \Pi_c$.
\end{defn}
\end{framed}
In particular, any throughput-optimal policy is a solution to 
\begin{align}\label{eq:thoptformulation}
\begin{aligned}
    \max_{\pi \in \Pi_c} &\min_{e \in \mathcal{T}} \bar{\mu}_e^{\pi} w_e \\
    \text{s.t.} \ &\mu^{\pi}(t) \in \mathcal{M}^{\phi} , \ \forall \ 0 \leq t \leq K^{\pi},
\end{aligned}
\end{align}
and we denote the solution as $\lambda^*(\boldsymbol{w})$. This has a closed-form solution given by the following.

\begin{framed}
\begin{theorem}\label{th:thoptsoln}
    For any set of slices $\boldsymbol{w}$ and interference model $\phi \in \Phi$,
    \begin{equation}\label{eq:thopteq}
        \lambda^*(\boldsymbol{w}) = \frac{1}{\phi+1} \min_{1 \leq j \leq |\mathcal{T}|-\phi} H(w_j,\dots,w_{j+\phi})
    \end{equation}
    where $H(\cdot)$ represents the harmonic mean.
\end{theorem}
\end{framed}

\begin{proof}
    The throughput-optimal formulation~\eqref{eq:thoptformulation} can be reformulated as 
    \begin{align}
    \begin{aligned}
        \max \ &\lambda(\boldsymbol{w}) \\
        \text{s.t.} \ &\lambda(\boldsymbol{w}) \leq \bar{\mu}_e w_e, \ \forall \ e \in \mathcal{T}, \\
        &\sum_{l=j}^{j+\phi} \bar{\mu}_l \leq 1, \ \forall \ 1 \leq j \leq |\mathcal{T}|-\phi,
    \end{aligned}
    \end{align}
    where link $j$ is understood to mean the $j$-th hop on the route $\mathcal{T}_j$. The second constraint is equivalent to the interference constraint in~\eqref{eq:thoptformulation} because only one of each set of $\phi+1$ adjacent links can be activated at once. Combining the two constraints yields 
    \begin{equation}
        \lambda(\boldsymbol{w}) \sum_{l=j}^{j+\phi} \frac{1}{w_l} \leq 1, \ \forall \ 1 \leq j \leq |\mathcal{T}|-\phi,
    \end{equation}
    and rearranging the equation yields
    \begin{equation}
        \lambda(\boldsymbol{w}) \leq \frac{1}{\phi+1} H(w_j,\dots,w_{j+\phi}), \ \forall \ 1 \leq j \leq |\mathcal{T}|-\phi.
    \end{equation}
    The maximum $\lambda^*(\boldsymbol{w})$ is then equal to the minimum over $j$, which completes the proof.
\end{proof}

This result shows that, for a fixed set of slices, the maximum supported rate scales as $1/\phi$, which approximates the average activation rate of each link. The harmonic mean of a set is dominated by the minimum, so this result also shows that the smallest slice (i.e., the bottleneck) has the largest impact on achievable rates, which matches our intuition.

\subsection{Deadline Optimality}

We next turn to characterizing feasible deadline guarantees. Define $\tau^*(\pi, \boldsymbol{w}, \lambda)$ as the minimum deadline a flow can meet, with rate $\lambda$ and slice widths $\boldsymbol{w}$, and under a policy $\pi$. Equivalently, this quantity is the largest delay seen by any packet in the flow, so we will sometimes refer to $\tau^*$ as maximum delay when appropriate. We define a deadline-minimizing policy for these parameters as follows.
\begin{framed}
\begin{defn}
    A policy $\pi \in \Pi_c$ is deadline-minimizing for rate $\lambda$ and slice widths $\boldsymbol{w}$ if $\tau^*(\pi, \boldsymbol{w}, \lambda) \leq \tau^*(\pi', \boldsymbol{w}, \lambda)$ for all $\pi' \in \Pi_c$.
\end{defn}
\end{framed}

There is clearly a relationship between $\tau^*$ and $\lambda$, because as rates increase, it becomes harder for every packet to meet its deadline. To capture the inherent scheduling delay caused by interference, we define the rate-independent quantity
\begin{equation}
    \tau^*(\pi) \triangleq \lim_{\lambda \to 0} \tau^*(\pi, \boldsymbol{w}, \lambda),
\end{equation}
and note that this is the smallest deadline that $\pi$ can guarantee for any $\lambda > 0$ and any set of slices. This allows us to define a \textit{deadline-optimal} policy, independent of arrival rate and slices.
\begin{framed}
\begin{defn}
    A policy $\pi \in \Pi_c$ is deadline-optimal if $\tau^*(\pi) \leq \tau^*(\pi')$ for all $\pi' \in \Pi_c$.
\end{defn}
\end{framed}

To avoid confusion, we note the distinction between the terms \textit{deadline-minimizing} when speaking in terms of a specific arrival rate, and \textit{deadline-optimal} when speaking independently of arrival rates. 

In the example in the previous section, we saw that schedule order plays an important role in making deadline guarantees, and we begin to formalize this here with the concept of inter-scheduling times. Denote the set of time slots where link $e$ is scheduled under a policy $\pi$ as  $\pi(e) \triangleq \{ 0 \leq t \leq K^{\pi} \ | \ \mu_e^{\pi}(t) = 1 \}$. Then define the minimum inter-scheduling time of links $e$ and $e'$ to be the smallest time interval between consecutive scheduling events of links $e$ and $e'$. Formally, we define this as
\begin{equation}
    \underline{k}_{e,e'}^{\pi} \triangleq \min_{t_1 \in \pi(e)} \min_{t_2 \in \pi(e')} (t_2 - t_1) \bmod K^{\pi},
\end{equation}
and similarly define the maximum inter-scheduling time as
\begin{equation}
    \overline{k}_{e,e'}^{\pi} \triangleq \max_{t_1 \in \pi(e)} \min_{t_2 \in \pi(e')} (t_2 - t_1) \bmod K^{\pi}.
\end{equation}

For ease of notation, we denote the min and max inter-scheduling times of a link $e$ with itself as $\underline{k}_e^{\pi}$ and $\overline{k}_e^{\pi}$ respectively. Using these quantities, we can bound $\tau_i^*(\pi)$ as follows.

\begin{framed}
\begin{lemma}\label{lemma:taulowerbound}
    For any admissible policy $\pi \in \Pi_c$,
    \begin{multline}\label{eq:taulowerbound}
        \underline{k}_1^{\pi} + \sum_{1 \leq j < |\mathcal{T}|} \underline{k}_{j,j+1}^{\pi} \leq \tau_i^*(\pi) \\ \leq \overline{k}_1^{\pi} + \sum_{1 \leq j < |\mathcal{T}|} \overline{k}_{j,j+1}^{\pi},
    \end{multline}
    where we slightly abuse notation to denote $\mathcal{T}_j$ as link $j$.
\end{lemma}
\end{framed}

\begin{proof}
    We first show the lower bound. Any packet delivered to its destination must traverse all links on its route in order, and under a policy $\pi$, every packet served by link $e$ that arrives at link $e+1$ must wait at least $\underline{k}_{e,e+1}^{\pi}$ slots before being served. Therefore, the smallest amount of time between being served at the source link and being served at the destination link is $\sum_{1 \leq j < |\mathcal{T}|} \underline{k}_{j,j+1}^{\pi}$, and packets must wait $\underline{k}_1^{\pi}$ slots from when they arrive at the source link until they are served.

    Similarly, because $\lambda \to 0$, all queued packets are served each time a link is scheduled, so a packet must wait at most $\overline{k}_1^{\pi}$ slots to be served at the source link and $\sum_{1 \leq j < |\mathcal{T}|} \overline{k}_{j,j+1}^{\pi}$ slots from when it is served at the source link to when it is served at the destination link. The result follows.
\end{proof}

These bounds show that in order to keep packet delays small, we should minimize inter-scheduling times between consecutive links on a route. We will show that deadline optimality is in fact achieved by a specific \textit{Ordered Round-Robin} ($ORR$) policy, under which the upper and lower bounds in Lemma~\ref{lemma:taulowerbound} are tight and minimized over all admissible policies.

Recall that under an interference model $\phi$, links separated by fewer than $\phi$ hops cannot be scheduled in the same slot. Then define the $ORR$ policy as follows. At each time $t$, we activate link $\mathcal{T}_j$, where $j=t \bmod \phi$, along with every $\phi+1$ subsequent links. Then at each slot, links at equally spaced intervals of $\phi+1$ hops are activated, beginning with links $\{0, \phi+1, 2 (\phi+1),\dots \}$ at time $t=0$, links $\{1, 1+(\phi+1), 1+2 (\phi+1), \dots \}$ at $t=1$, and so on, where again link $j$ refers to $\mathcal{T}_j$. Note that this schedule has a period $K^{ORR} = \phi+1$.

\begin{framed}
\begin{theorem}\label{th:orrbound}
    The $ORR$ policy is deadline-optimal under any interference model $\phi \in \Phi$, with a maximum packet delay
    \begin{equation}
        \tau^*(ORR) = |\mathcal{T}|+\phi,
    \end{equation}
    and activation rates
    \begin{equation}
        \bar{\mu}_e^{ORR} = \frac{1}{\phi+1}, \ \forall \ e \in \mathcal{T}.
    \end{equation}
\end{theorem}
\end{framed}

\begin{proof}
    See Appendix~\ref{app:orrproof}.
\end{proof}

Note that under the $ORR$ policy, the bounds in~\eqref{eq:taulowerbound} are tight because inter-scheduling times are equal for all individual links and link pairs, with $\underline{k}_1 = \overline{k}_1 = \phi$ and $\underline{k}_{j,j+1} = \overline{k}_{j,j+1} = 1$ for all $j$.

From~\eqref{eq:lambdabound}, the maximum rate the $ORR$ policy can support for a set of slices $\boldsymbol{w}$ is
\begin{equation}
    \lambda^*(ORR, \boldsymbol{w}) = \frac{1}{\phi+1} \min_{e \in \mathcal{T}} w_e,
\end{equation}
which is not throughput-optimal in general. From Theorem~\ref{th:thoptsoln}, the difference between this and the throughput-optimal rate is
\begin{multline}\label{eq:orrlambdaloss}
    \lambda^*(\boldsymbol{w}) - \lambda^*(ORR, \boldsymbol{w}) \\= \frac{1}{\phi+1} \big( \min_{1 \leq j \leq |\mathcal{T}|-\phi} H(w_j,\dots,w_{j+\phi}) - \min_{e \in \mathcal{T}} w_e \big),
\end{multline}
where we again slightly abuse notation to denote $\mathcal{T}_j$ as link $j$. We can immediately observe that when all slice widths are equal, this quantity becomes zero.

\begin{framed}
\begin{corollary}\label{cor:jointoptimalpoint}
    When slice widths are equal across a flow's route, the $ORR$ policy is both throughput-optimal and deadline-optimal.
\end{corollary}
\end{framed}

Under general slice widths, however,~\eqref{eq:orrlambdaloss} exhibits a tradeoff between supported rates and deadlines. We can achieve deadline optimality with the $ORR$ policy at a cost of sacrificing throughput, and similarly, we can choose an alternative policy to $ORR$ and achieve higher throughput at a cost of sacrificing deadlines. From Corollary~\ref{cor:deadlinecondition}, a policy under the throughput-optimal rate $\lambda^*(\boldsymbol{w})$ meets deadline $\tau$ if and only if
\begin{equation}
    \frac{1}{\lambda^*(\boldsymbol{w})} \max_{0 \leq t \leq T} \sum_{e \in \mathcal{T}} Q_e^{\pi}(t) \leq \tau_i.
\end{equation}
Then the smallest achievable deadline under a throughput-optimal policy is the solution to
\begin{multline}\label{eq:mintau}
    \min_{\pi} \tau^* \big(\pi, \boldsymbol{w}, \lambda^*(\boldsymbol{w}) \big) \\
    \begin{aligned}
        &= \min_{\pi} \frac{1}{\lambda^*(\boldsymbol{w})} \max_{0 \leq t \leq T} \sum_{e \in \mathcal{T}} Q_e^{\pi}(t) \\
        &= \min_{\pi} \frac{\phi+1}{\min_j H(w_j,\dots,w_{j+\phi})} \max_{0 \leq t \leq T} \sum_{e \in \mathcal{T}} Q_e^{\pi}(t) \\
        &= \min_{\pi} \Gamma(\pi,\boldsymbol{w}) \cdot (\phi+1) \cdot |\mathcal{T}|,
    \end{aligned}
\end{multline}
within a value of $1$ because $\tau^*$ must be an integer, and where we define
\begin{equation}
    \Gamma(\pi,\boldsymbol{w}) \triangleq \frac{\max_{0 \leq t \leq T} \frac{1}{|\mathcal{T}|} \sum_{e \in \mathcal{T}} Q_e^{\pi}(t)}{\min_j H(w_j,\dots,w_{j+\phi})}.
\end{equation}
The min is taken over all throughput-optimal $\pi$ which support rate $\lambda^*(\boldsymbol{w})$. Note that this is less than a factor of $\Gamma(\pi,\boldsymbol{w})(\phi+1)$ from $\tau^*(ORR)$.

Minimizing $\Gamma(\pi,\boldsymbol{w})$ over $\pi$, however, involves optimizing over all feasible schedule orders which support a rate $\lambda^*(\boldsymbol{w})$. This is a combinatorial optimization problem with a state space that grows exponentially in $|\mathcal{T}|$. In fact, this characterization shows that solving $\min_{\pi} \tau^*(\pi,\boldsymbol{w},\lambda)$ for \textit{any} rate $\lambda > \lambda^*(ORR,\boldsymbol{w})$ has worst-case exponential complexity. This illustrates the difficulty of minimizing delay even in a line network with a single flow.

\subsection{Bottleneck Slice Widths}

One factor that adds to the complexity of~\eqref{eq:mintau} is that a link can serve a variable number of packets each time it is activated, depending on the current queue size. We can simplify analysis by considering a special case where every link is a bottleneck. For a given $\lambda$ and set of slices $\boldsymbol{w}$, we know that $\bar{\mu}_e \geq \frac{\lambda}{w_e}$ for queues to be stable, and that $\sum_{j=0}^{\phi} \bar{\mu}_{e+j} \leq 1$ for all $0 \leq e < |\mathcal{T}|-\phi$ due to interference. Combining these conditions, the minimum slice widths $\boldsymbol{w'} \leq \boldsymbol{w}$ that support $\lambda$ are the solution to the convex program
\begin{align}
\begin{aligned}
    \min_{\boldsymbol{w'}} \ &\sum_{e \in \mathcal{T}} w'_e \\
    \text{s.t.} \ &\sum_{j=0}^{\phi} \frac{1}{w'_{e+j}} \leq \frac{1}{\lambda}, \ \forall \ 0 \leq e < |\mathcal{T}|-\phi, \\
        &0 \leq w'_{i,e} \leq w_e, \ \forall \ e \in \mathcal{T}.
\end{aligned}
\end{align}

Consider an augmented network with slice widths equal to $\boldsymbol{w'}$. We refer to these as bottleneck slice widths because no slice can be reduced while still supporting the rate $\lambda$. Because each slice $w'_e \leq w_e$, the minimum achievable deadline $\tau^*(\pi, \boldsymbol{w'}, \lambda)$ is an upper bound on the original $\tau^*(\pi, \boldsymbol{w}, \lambda)$ under any policy $\pi$, because no more packets can ever be served with slices $\boldsymbol{w'}$. Thus the number of packets in the network at a given time can only be larger, and because this holds for any $\pi$, it must hold for the deadline-minimizing policy.

Define a greedy scheduling policy under bottleneck slices as follows. At each time $t$, schedule a link $e$ if $Q_e(t) \geq w'_e$, and if more than one such queue exists, schedule the one closest to the destination. If no such queues exist, schedule the source link. This policy is optimal under total interference.

\begin{framed}
\begin{corollary}\label{cor:greedybound}
    Under total interference and bottleneck slice widths $\boldsymbol{w'}$, greedy is deadline-minimizing for any supported rate $\lambda$, with
    \begin{equation}\label{eq:taugreedy}
        \tau^*(greedy,\boldsymbol{w'},\lambda) \lesssim \frac{1}{\lambda} \sum_{e \in \mathcal{T}} w'_e = \sum_{e \in \mathcal{T}} \frac{1}{\bar{\mu}_e^{greedy}}.
    \end{equation}
\end{corollary}
\end{framed}

\begin{proof}
    See Appendix~\ref{app:greedyboundproof}.
\end{proof}

Greedy can be extended to any value of $\phi$ by again starting at the destination link and moving toward the source, scheduling the first link observed with a full queue. When a link is scheduled, the subsequent $\phi$ links are skipped, and the next link with a full queue is scheduled. This is repeated until the source link is reached. Unfortunately, this is not guaranteed to be optimal because more than one link can be scheduled per slot. Scheduling the greedy links may preclude another link or set of links from being scheduled which would improve future performance. For this reason, it is unlikely that any myopic policy is optimal outside of the total interference model. Nevertheless, greedy still prevents queues from growing too large and prioritizes packets closer to the destination, enabling them to be delivered quickly. As a result, we conjecture that greedy is near-optimal for any $\phi$, and a good heuristic for minimizing $\tau^*$ under bottleneck slice widths.

When slice widths are equal, the greedy policy reduces to $ORR$ as we expect. This follows because if link $\mathcal{T}_j$ is scheduled at time $t$ under greedy, then link $\mathcal{T}_{j+1}$ will be scheduled at $t+1$ because slice widths are equal, so $Q_{j+1}$ must be full. By induction, this policy becomes $ORR$. Intuitively, the more disparate the slice widths are, the more burstiness and delay packets will see, as evidenced by the bound in~\eqref{eq:taugreedy}. A particularly large slice width will lead to a small activation rate and therefore a large backlog under greedy, causing the total delay to increase. 

\section{General Feasibility}\label{sec:generalfeasibility}

We now turn to understanding general feasibility for simultaneous flows with unique source/destination pairs. We again attempt to characterize the feasible guarantees that can be made under any policy in $\Pi_c$, this time applying guarantees to all flows jointly. While slicing decouples the queueing dynamics between flows, we have seen that meeting tight deadlines is dependent on schedule order, which couples flows together through the scheduling policy.

We assume that multiple flows can be routed over the same link, and our policy is able to choose the slice width allocated to each flow, subject to capacity constraints. In particular, we must jointly optimize slice widths and scheduling to satisfy deadline constraints following the condition in Corollary~\ref{cor:deadlinecondition} and capacity constraints following 
\begin{equation}
    \sum_{{f_i: e \in \mathcal{T}^{(i)}}} w_{i,e} \leq c_e, \ \forall \ e \in E.
\end{equation}
Combining this with the stability condition in~\eqref{eq:ratestabcond} yields 
\begin{equation}
    \sum_{{f_i: e \in \mathcal{T}^{(i)}}} \frac{\lambda_i}{\bar{\mu}_e^{\pi}} \leq c_e, \ \forall \ e \in E,
\end{equation}
which is a necessary and sufficient condition for stability under $\pi \in \Pi_c$, by implicitly setting slice widths equal to $w_{i,e} = \lambda_i \ \bar{\mu}_e^{\pi}$. Recall that slices of this size are by definition bottleneck slice widths.

Extending notation from the previous section, we denote $\tau_i^*(\pi, \lambda_i)$ as the minimum achievable deadline, or equivalently the maximum packet delay, for a flow $f_i$ with rate $\lambda_i$ under a policy $\pi$ which chooses both the schedule and slice widths. In the previous section, we saw that finding a deadline-minimizing policy for a solitary flow has complexity that grows exponentially in the number of queues $|\mathcal{T}|$. This remains true in the general case, where the number of queues is $\sum_{f_i \in \mathcal{F}} |\mathcal{T}^{(i)}|$, and the worst-case complexity grows exponentially in this quantity. This problem is intractable for all but the smallest problem sizes. Furthermore, because the direction of traffic is no longer the same for all flows, and interference is no longer limited to a line network, the $ORR$ and greedy policies developed in the previous section are clearly no longer applicable.

There has been considerable work on scheduling policies that ensure queue stability under interference constraints. Before attempting to optimize for deadlines directly, it is instructive to understand how large packet delays can become under a policy which guarantees that queues remain stable and thus bounded. Because queues are bounded, the maximum packet delay is also bounded, and we can derive an exact expression.

\begin{framed}
    \begin{theorem}\label{th:deadlinebound}
        The worst-case packet delay under any stabilizing policy in $\Pi_c$ is given by
        \begin{equation}\label{eq:deadlinebound}
            \max_{\pi \in \Pi_c} \tau_i^*(\pi, \lambda_i) = \max_{\pi \in \Pi_c} K^{\pi} \sum_{e \in \mathcal{T}^{(i)}} (1 - \bar{\mu}_e^{\pi}).
        \end{equation}
        Moreover, there exist many such policies where $\tau_i^*$ grows linearly with $K^{\pi}$.
    \end{theorem}
    \end{framed}
    
\begin{proof}
    See Appendix~\ref{app:deadlineboundproof}.
\end{proof}
    
Recall that $\bar{\mu}_e^{\pi} = \lambda_i / w_{i,e} = \eta_e^{\pi} / K^{\pi}$, and after rearranging,
\begin{equation}
    K^{\pi} = \eta_e^{\pi} \frac{w_{i,e}}{\lambda_i}, \ \forall \ i, e.
\end{equation}
In order to be integer, $K^{\pi}$ must be an integer multiple of the denominator of $w_{i,e} / \lambda_i$, for all $i$ and $e$. This is a rational quantity by assumption, but the denominator can grow aribtrarily large under our fluid traffic model. Even when $\lambda_i$ is constrained to be integer and represent a discrete number of packets, $K^{\pi}$ can be as large as $\prod_i \lambda_i$.

This shows that many policies in $\Pi_c$ are insufficient to meet tight deadline guarantees. In fact, it shows that there exists at least one policy where packets experience delays within a constant factor of $K^{\pi} |\mathcal{T}^{(i)}|$, and many policies where packet delays grow with $K^{\pi}$. These are not pathological examples, but rather any policy with inter-scheduling times that are a function of the scheduling period. This includes frame-based policies, like much of the existing QoS scheduling literature, or simply random orderings of activations which happen to have large inter-scheduling times.

To avoid delays that grow with $K^{\pi}$, and motivated by the complexity of searching over all schedule orders, we derive alternative deadline conditions which drastically reduce the state space of the problem. Using a similar argument to that in Lemma~\ref{lemma:taulowerbound}, which gives a bound on packet delays in the solitary flow setting, we will show an upper bound on packet delays in our general case subject to conditions on slice widths and maximum inter-scheduling times.

\subsection{Delay Deficit}

First we introduce a quantity called delay deficit. This quantity tracks how long a packet has been in the network relative to a pre-defined quota, which dictates the delay a packet should ``expect'' to see at each link. Consider a packet $l$ belonging to flow $f_i$, and denote the time that it arrived at its source node as $t_0^{l}$ and the age of this packet at time $t$ as $a^l(t) \triangleq t - t_0^l$. We will show that by setting the delay quota for link $e$ under policy $\pi \in \Pi_c$ to the maximum inter-scheduling time $\overline{k}_e^{\pi}$, the delay deficit of any packet will never grow too large. For ease of notation, we will drop the overline in the quantity $\overline{k}_e^{\pi}$ for the remainder of the paper and refer to this simply as $k_e^{\pi}$.

Define the delay deficit of packet $l$ at time $t$ as 
\begin{equation}\label{eq:delaydeficit}
    \delta^l(t) \triangleq a^l(t) - \sum_{e \in \tilde{T}^{(l)}} k_e^{\pi},
\end{equation}
where $\tilde{T}^{(l)}$ denotes all links on packet $l$'s route where it has already been served. Delay deficit evolves in the following way. When packet $l$ arrives at the source, $\delta^l(t^l_0) = 0$. Then it is incremented by one in each subsequent slot until it reaches its destination, and it is decremented by $k_e^{\pi}$ whenever it is served at link $e$. A negative delay deficit signifies that a packet has seen less delay than its quota on its route thus far, making it ahead of schedule. A delay deficit larger than the delay quota at a packet's current link signifies that the packet is behind schedule.

\begin{framed}
\begin{lemma}\label{lemma:delaydeficit}
    Under any policy $\pi \in \Pi_c$, with maximum inter-scheduling times $k_e^{\pi}$ and slice widths $w_{i,e} \geq \lambda_i k_e^{\pi}$, for all $e \in \mathcal{T}^{(i)}$, the delay deficit of a packet belonging to flow $f_i$ will never be larger than zero when it arrives at a new link.
\end{lemma}
\end{framed}

\begin{proof}
    See Appendix~\ref{app:delaydeficit}.
\end{proof}

It is important to note that this result and the subsequent theorem hold even though a packet is not necessarily served at link $e$ within $k_e^{\pi}$ slots of arrival. This would require larger slice widths to guarantee that queues are emptied each time a link is scheduled. Rather, our choice of slice widths ensures that a packet is scheduled within $k_e^{\pi}$ slots of when its delay deficit becomes positive.

This is demonstrated by the example in Figure~\ref{fig:delay-deficit}, which shows the first two queues of a flow $f_i$'s route and the delay deficit of one packet. For simplicity we let $\lambda_i = 1$ and assume that $k_0 = 3$ and $k_1 = 4$, noting that the schedule shown in the figure satisfies these maximum inter-scheduling times. Let slice widths $w_{i,0} = 3$ and $w_{i,1} = 4$, which satisfy the condition in Lemma~\ref{lemma:delaydeficit}. The figure shows the queue position and delay deficit of the packet marked with an X. Note that the packet remains in $Q_{i,1}$ for longer than $k_1$ slots, but that its delay deficit is still negative when it arrives at the next link. This is possible because the packet only spends one slot in $Q_{i,0}$, so it arrives at the second link ``early'' with a negative delay deficit.

An immediate consequence of Lemma~\ref{lemma:delaydeficit} is the following end-to-end delay bound, which will form the basis for our scheduling policies in the next section.

\begin{figure}
    \centering
    \includegraphics[width=0.4\textwidth]{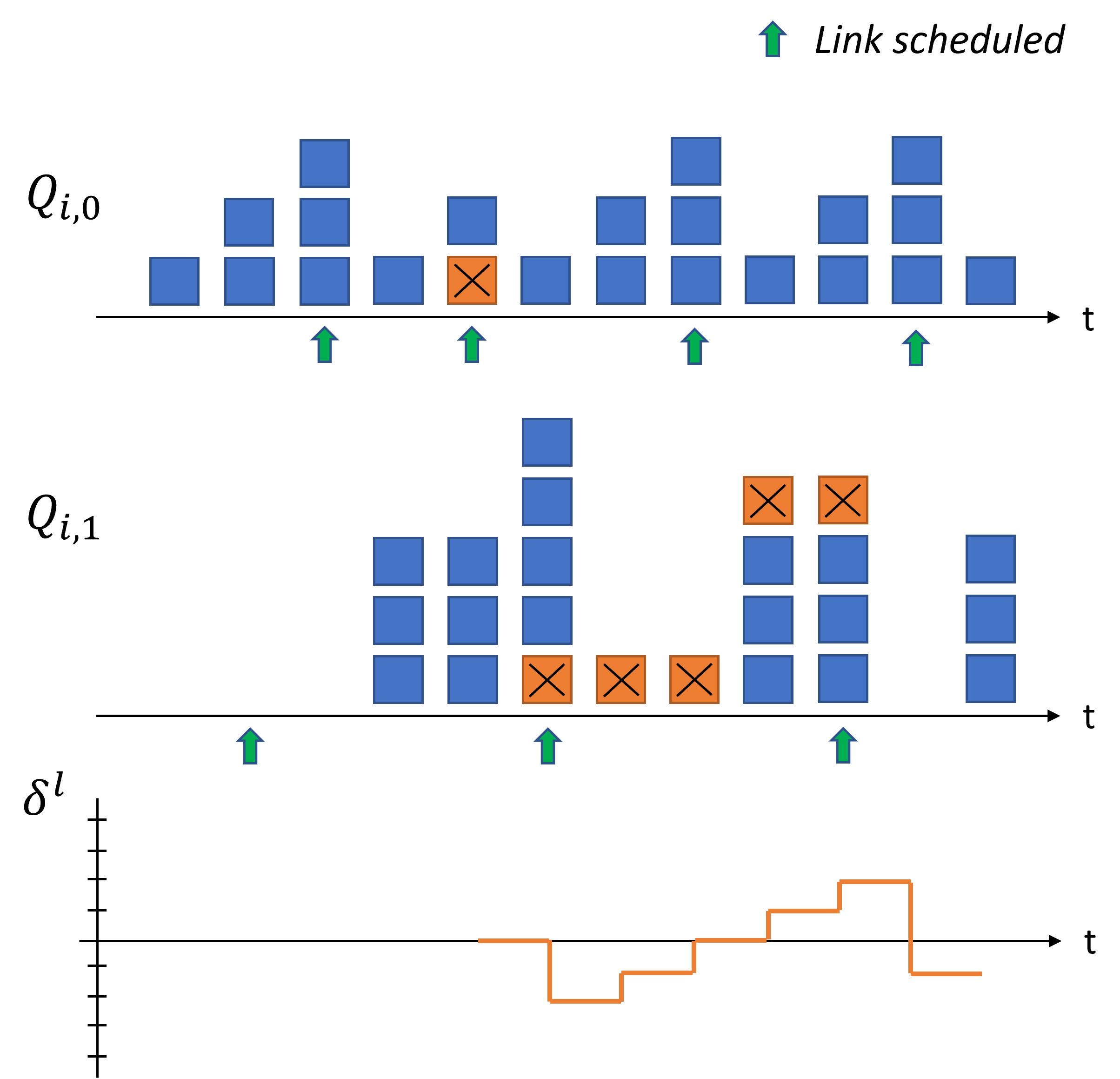}
    \caption{Queue Size and Delay Deficit Evolution}
    \label{fig:delay-deficit}
\end{figure}

\begin{framed}
\begin{theorem}\label{th:maindelaydefbound}
    Under any policy $\pi \in \Pi_c$, with maximum inter-scheduling times $k_e^{\pi}$ and slice widths $w_{i,e} \geq \lambda_i k_e^{\pi}$, for all $e \in \mathcal{T}^{(i)}$,
    \begin{equation}\label{eq:tauupperbound}
        |\mathcal{T}^{(i)}| \leq \tau_i^*(\pi,\lambda_i) \leq \sum_{e \in \mathcal{T}^{(i)}} k_e^{\pi}.
    \end{equation}
\end{theorem}
\end{framed}

\begin{proof}
    The lower bound is simply the length of the route, so it holds trivially. Now we will show the upper bound. Imagine a ficticious exit link leaving the destination node of $f_i$. From Lemma~\ref{lemma:delaydeficit}, the delay deficit of a packet which reaches this link is at most zero, so from~\eqref{eq:delaydeficit} the packet's age is at most $\sum_{e \in \mathcal{T}^{(i)}} k_e^{\pi}$. This completes the proof.
\end{proof}

We note several things about this result. First, it subsumes the worst-case delay bound in Theorem~\ref{th:deadlinebound}. In particular, in the proof of Theorem~\ref{th:deadlinebound}, the worst-case scenario is shown to have a maximum inter-scheduling time $k_e^{\pi} = K^{\pi}(1-\bar{\mu}_e^{\pi})$ for all links $e$, which makes the bounds identical under this policy. Second, when links are scheduled more regularly and $k_e^{\pi}$ is independent of the schedule length, this bound can be significantly tighter than that in Theorem~\ref{th:deadlinebound} and is within a factor of the average $k_e^{\pi}$ from the lower bound.

Intuitively, bounding the inter-scheduling times for each link bounds how far a policy can deviate from the $ORR$ policy for each flow. The $ORR$ policy ensures that once packets are served at a link, they are immediately served at the subsequent link in the next slot. While this is no longer feasible in the general case, the delay deficit conditions show that no packet can wait too long to be served when inter-scheduling times are small, bounding the deviation from the $ORR$ policy.

A final thing to note is that, while the necessary conditions on slice widths are larger than their bottleneck values, the amount of additional slice capacity required is generally small. Define this additional capacity for slice $w_{i,e}$ as 
\begin{equation}\label{eq:deltaw}
    \Delta w_{i,e}(\pi) \triangleq \lambda_i \Bigl ( k_e^{\pi} - \frac{1}{\bar{\mu}_e^{\pi}} \Bigr ) \geq 0.
\end{equation}
Because the average inter-scheduling time is $1/\bar{\mu}_e^{\pi}$, $\Delta w_{i,e}(\pi)$ is small when inter-scheduling times are somewhat regular and the maximum is not much larger than the average. If inter-scheduling times are always equal, then $k_e^{\pi} = 1 / \bar{\mu}_e^{\pi}$ for all $e \in E$, and $\Delta w_{i,e}(\pi) = 0$. When this holds we say that $\pi$ is a \textit{regular schedule}.

\begin{framed}
\begin{corollary}\label{cor:regularschedules}
    A policy $\pi \in \Pi_c$ with maximum inter-scheduling times $\boldsymbol{k^{\pi}}$ supports a set of flows $\mathcal{F}$ if it is a solution to
    \begin{align}\label{eq:rateconditions}
    \begin{aligned}
       (P1) \ : \ \text{find} \ &\pi \\
    \text{s.t.} \ &\sum_{e \in \mathcal{T}^{(i)}} k_e^{\pi} \leq \tau_i, \ \forall \ f_i \in \mathcal{F}, \\
    &\sum_{f_i : e \in \mathcal{T}^{(i)}} w_{i,e} \leq c_e, \ \forall \ e \in E, \\
    &w_{i,e} = \lambda_i k_e^{\pi}, \ \forall \ f_i \in \mathcal{F}, \ e \in \mathcal{T}^{(i)}, \\
    &\mu^{\pi}(t) \in \mathcal{M}, \ \forall \ 0 \leq t < K^{\pi}.
    \end{aligned}
    \end{align}
    Furthermore, if $\pi$ is a regular schedule, then all slices are equal to their bottleneck values.
\end{corollary}
\end{framed}

\begin{proof}
    The sufficient conditions come from the deadline bound in Theorem~\ref{th:maindelaydefbound} and link capacity constraints. If both of these are satisfied for all flows and links, then $\pi$ is guaranteed to support $\mathcal{F}$ by Theorem~\ref{th:maindelaydefbound}.
\end{proof}

The conditions in Corollary~\ref{cor:regularschedules} are sufficient but not necessary for a policy to support $\mathcal{F}$. It is interesting to note, however, that after combining constraints, the deadline bound becomes
\begin{equation}
    \sum_{e \in \mathcal{T}^{(i)}} k_e = \frac{1}{\lambda_i} \sum_{e \in \mathcal{T}^{(i)}} w_{i,e} \leq \tau_i,
\end{equation}
which is equal to the approximate bound on $\tau^*$ under the greedy policy for a solitary flow in~\eqref{eq:taugreedy}. Recall that greedy is optimal under total interference, and we conjecture that it is near-optimal under more relaxed interference constraints. This leads us to conclude that the conditions in Corollary~\ref{cor:regularschedules} are close to necessary and further motivates our approach.

\subsection{Pinwheel Scheduling}

The conditions in Corollary~\ref{cor:regularschedules} present a massive reduction in complexity from tracking the age of each individual packet to ensure its deadline is met. Nevertheless, finding a schedule which satisfies these conditions is still a non-trivial task. In addition to finding a vector of integers $\boldsymbol{k^{\pi}}$ which satisfies constraints, we must also construct a feasible schedule that satisfies these maximum inter-scheduling times. In particular, there must exist a schedule of links such that no two consecutive appearances of link $e$ in the schedule are farther than $k_e^{\pi}$ slots apart. This type of scheduling has been studied under the name pinwheel scheduling~\cite{holte1989pinwheel,holte1992pinwheel} and periodic maintenance scheduling~\cite{wei_periodic_1983} with somewhat limited success. Under a total interference model where only one link can be scheduled at a time, determining whether a schedule exists for a given vector $\boldsymbol{k^{\pi}}$ is known to be NP-hard~\cite{mok_algorithms_1989,bar-noy_minimizing_2002}.

To illustrate the complexity of the problem, we assume a total interference model and a fixed vector $\boldsymbol{k}$ that satisfies~\eqref{eq:rateconditions} (for ease of notation, we drop the superscript $\pi$). Following the nomenclature in the literature, define
\begin{equation}
    \rho(\boldsymbol{k}) \triangleq \sum_{e \in E} \frac{1}{k_e}
\end{equation}
as the \textit{density} of a scheduling vector $\boldsymbol{k}$. Because $k_e$ represents the maximum inter-scheduling time, it is larger than the average inter-scheduling time $1/\bar{\mu}_e^{\pi}$. Therefore, $\rho(\boldsymbol{k}) \leq \sum_{e \in E} \bar{\mu}_e$, and a necessary condition for feasibility is that $\rho(\boldsymbol{k}) \leq 1$. This is trivially satisfied in some cases like round-robin. The vector $\boldsymbol{k} = (3,3,3)$ for example has density $1$ and is satisfied by the schedule $\{0,1,2\}$, where recall that all schedules are by definition cyclic and repeating. The vector $\boldsymbol{k} = (2,4,4)$ likewise has density $1$ and is schedulable with $\{0,1,0,2\}$.

The density condition is far from sufficient, however. Consider a vector $\boldsymbol{k} = (2,3,x)$. It is easy to verify that no matter how large a value we assign to $x$, no satisfying schedule exists for this vector. Because this vector has a density of $5/6 + 1/x$, there is no guarantee that a schedule exists for any vector with density larger than $5/6$. It has long been conjectured that all vectors with density up to $5/6$ are schedulable, and this was recently proven to be true~\cite{kawamura_proof_2024}, though scheduling all such vectors in polynomial time remains elusive.

Furthermore, density is not the only factor that affects schedulability. Consider the vector $\boldsymbol{k} = (3,4,8,8,8)$ with density $0.96$. This is satisfied by the schedule $\{0,1,2,0,3,1,0,4\}$. Now consider the vector $\boldsymbol{k} = (3,5,7,8,8)$, with density $0.93$. While this vector appears similar and in fact has a lower density, there is no way to construct a satisfying schedule, and enumerating over all possible schedules in an effort to do so is intractable for all but the shortest schedule lengths.

There are certain properties which make vectors easy to schedule, however, and algorithms exist to schedule these vectors in polynomial time. These include all vectors with density up to $1$ that contain at most two distinct values~\cite{holte1992pinwheel}, such as $\boldsymbol{k} = (4,4,6,6,6)$, and all vectors with density up to $1$ where each value in the sorted vector is a multiple of the previous~\cite{holte1989pinwheel,chan_schedulers_1993}, such as $\boldsymbol{k} = (2,6,6,12,12)$. We refer to the second class of vectors as step-down vectors following~\cite{li_scheduling_2021}.

There are also a number of algorithms which apply rounding techniques to force a vector to satisfy one of the above conditions. Because values can always be rounded down without violating the inter-scheduling times in the original vector, rounding values down until the resulting vector is step-down, for example, is a valid technique. The most comprehensive of these algorithms is known as \textit{Scheduler xy}, or $S_{xy}$, and can schedule all vectors with density up to $0.7$~\cite{chan_general_1992}, plus a subset of vectors with densities larger than this value (including the three cases above). This is the closest density guarantee to the $5/6$ bound of any known algorithm.

To further complicate things, the discussion thus far only applies to the total interference model, where every link is mutually interfering. Under more relaxed interference models, multiple links can be scheduled per slot, which creates a more general version of the problem. In this generalized version, a schedule must simultaneously satisfy both inter-scheduling times and interference constraints. We refer to the joint problem of forming independent sets and scheduling them under pinwheel constraints as the \textit{pinwheel coloring} problem.

\begin{framed}
\begin{theorem}
    $(P1)$ generalizes the pinwheel coloring problem, which is NP-hard.
\end{theorem}
\end{framed}

\begin{proof}
    We will show that $(P1)$ generalizes both pinwheel scheduling and graph coloring. Assume that each route has a length of one link so deadline constraints reduce to an upper bound on each $k_e$, and that capacities are infinite. Then, under a total interference model, the problem reduces to finding a schedule such that the inter-scheduling time of each task is bounded by some quantity, which is precisely the pinwheel scheduling problem.

    Now assume that all values of $\tau_i$ are equal, but the conflict graph allows more than one link to be scheduled at once. Then deciding whether every link can meet its inter-scheduling time is equivalent to deciding whether the conflict graph can be colored with $\tau_i$ colors, which is the vertex coloring problem. Both pinwheel scheduling and vertex coloring are NP-complete, so $(P1)$ is NP-hard.
\end{proof}

To generate schedules efficiently in light of this result, we develop an algorithm in the next section which leverages the sets of vectors $\boldsymbol{k}$ that can be scheduled in polynomial time, while efficiently handling interference constraints. We demonstrate the algorithm's performance through extensive simulations in Section~\ref{sec:simulations}.

\section{Algorithm Design}\label{sec:algorithmdesign}

In this section we design a polynomial-time algorithm which is able to solve many instances of $(P1)$ and the Pinwheel Coloring problem. Because the problem is NP-hard, it cannot find a supporting policy in all cases where one exists, but simulation results in the next section show that it can simultaneously achieve near-optimal throughput and tight deadlines in a variety of network scenarios.

Our algorithm consists of two stages, and for ease of exposition we present them using the conflict graph representation of the network, where each node $v \in V_c$ corresponds to a link in the network graph, and edges between nodes represent interference constraints. In the first stage, we construct activation sets which form a coloring of the conflict graph, so each node $v \in V_c$ belongs to exactly one set. At each time step, we will activate one of these sets. Denote the set of activation sets after coloring as $\mathcal{S}$, and the maximum inter-scheduling time of set $s \in \mathcal{S}$ as $k_s$. Then the maximum inter-scheduling time of a node $v$ in the conflict graph is equal to $k_s$ when $v \in s$.

After forming $\mathcal{S}$, we set $k_s$ such that it satisfies the conditions in~\eqref{eq:rateconditions} for all $v \in s$, and denote the vector of inter-scheduling times $k_s$ as $\boldsymbol{k_{\mathcal{S}}}$. Because vectors with lower density are generally easier to schedule, our objective in the first stage of the algorithm is to find a coloring $\mathcal{S}$ and corresponding $\boldsymbol{k_{\mathcal{S}}}$ that minimizes $\rho(\boldsymbol{k_{\mathcal{S}}})$. Explicitly, this is a solution to
\begin{align}
\begin{aligned}
    (P2) \ : \ \min_{\mathcal{S}, \boldsymbol{k_{\mathcal{S}}}} \ &\sum_{s \in \mathcal{S}} \frac{1}{k_s} \\
    \text{s.t.} \ &k_v = k_s, \ \forall \ v \in s, s \in \mathcal{S}, \\
    &\sum_{v \in \mathcal{T}^{(i)}} k_v \leq \tau_i, \ \forall \ f_i \in \mathcal{F}_d, \\
    &k_v \leq \frac{c_v}{\lambda_v} , \ \forall \ v \in V_c, \\
    &k_v \in \mathbb{Z}_+, \ \forall \ v \in V_c,
\end{aligned}
\end{align}
where we minimize $\mathcal{S}$ over the set of all valid colorings of $V_c$, and where $\lambda_v \triangleq \sum_ {f_i : v \in \mathcal{T}^{(i)}} \lambda_i$ is the total traffic routed over $v$. 

After solving $(P2)$, the second stage of the algorithm attempts to form a satisfying schedule using the $S_{xy}$ pinwheel scheduling algorithm. If the algorithm finds a satisfying schedule then it is guaranteed to support all flows in $\mathcal{F}$ from Corollary~\ref{cor:regularschedules}.

\subsection{Weighted Greedy Coloring}

We begin by addressing how to solve $(P2)$. This problem is NP-hard because if all values in $\boldsymbol{k_{\mathcal{S}}}$ are equal, it reduces to finding a minimum vertex coloring, which is a well known NP-complete problem. To solve it efficiently, we devise a heuristic based on a greedy coloring algorithm.

Let the vertices $V_c$ be ordered in some manner which we will fix later, and assume we have a set of colors $\{0,1,\dots\}$. For each vertex, the algorithm assigns it the smallest feasible color, i.e., the first color which does not already contain a conflicting vertex. It proceeds in order until each vertex is assigned a color. The motivation for this approach is twofold. First, it runs in $O(|V_c|)^2$ time and so satisfies our polynomial-time objective. And second, while greedy coloring can be far from optimal in general, it tends to perform well in our setting of a conflict graph based on local interference constraints. 

Greedy coloring has been studied specifically under $\phi$-hop interference to solve the maximum weighted independent set problem. When the connectivity graph is geometric, i.e., nodes are fixed on a plane and share a link if they are within some distance of each other, it has been shown to be a constant factor from optimal in the worst case~\cite{sharma2006complexity}. By iteratively solving this problem with equal weights, it is easy to see that greedy is a constant factor from optimal for the unweighted coloring problem. Similarly, on random geometric graphs with chromatic number $\chi$, it has been shown that greedy colors the graph with at most $\chi+1$ colors with high probability as the size of the network goes to infinity~\cite{McdiarmidColin2011Otcn}. In developing our greedy algorithm, we are motivated by connectivity graphs that can be represented as geometric graphs, or behave similarly. Simulation results in the next section verify the effectiveness of this choice on a variety of graph topologies.

An optimal coloring of any graph can be found by greedy when vertices are ordered correctly, simply by taking the optimal solution, ordering vertices by color, and running greedy on that order. On the other hand, there are pathological examples where greedy uses $|V_c|/2$ colors to color a $2$-colorable graph by using a bad ordering. Saying that greedy performs well on a graph effectively means that order does not matter, and that any arbitrary ordering yields approximately the same result for the unweighted coloring problem. However, recall that $(P2)$ is a \textit{weighted} coloring problem, where the weight of each color $s$ is $1/k_s = \max_{v \in s} 1/k_v$, which is tightly coupled to the coloring itself. Then even if the number of colors returned by greedy is largely independent of order, the weighted solution and density of the resulting vector $\boldsymbol{k_{\mathcal{S}}}$ may not be.

Clearly if we knew the coloring and the value of the optimal $k_v^*$ a priori for each $v$, then ordering vertices by $\boldsymbol{k^*}$ would be sufficient for greedy to recover the optimal solution. To approximate this, we solve a relaxation of $(P2)$ for an estimate $\boldsymbol{\hat{k}}$, and order vertices by these values in our greedy algorithm to form activation sets $\mathcal{S}$. Then with the coloring fixed, we solve the unrelaxed problem $(P2)$ for $\boldsymbol{k_{\mathcal{S}}}$.

Denote node $v$'s activation set after coloring as $s(v)$, and the size of $s(v)$ as $\omega_v$. Then the objective of $(P2)$ is equivalent to 
\begin{equation}\label{eq:p2equivalence}
    \sum_{s \in \mathcal{S}} \frac{1}{k_s} = \sum_{s \in \mathcal{S}} \frac{1}{|s|} \sum_{v \in s} \frac{1}{k_v} = \sum_{v \in V_c} \frac{1}{\omega_v k_v},
\end{equation}
where the first equality follows from the $(P2)$ constraint $k_v = k_s$, for all $v \in s$. This shows that the optimization only depends on the size of the activation sets $\omega_v$, and not on the coloring itself.

To approximate the optimal ordering for greedy, we replace the objective with this representation and begin by relaxing the $k_v = k_s$ constraint and the integrality condition on the vector $\boldsymbol{k_{\mathcal{S}}}$. This yields the program
\begin{align}\label{eq:relaxedgreedy}
\begin{aligned}
    \min_{\mathcal{S}, \boldsymbol{k}} \ &\sum_{v \in V_c} \frac{1}{\omega_v k_v} \\
    \text{s.t.} \ &\sum_{v \in \mathcal{T}^{(i)}} k_v \leq \tau_i, \ \forall \ f_i \in \mathcal{F}_d, \\
    &k_v \leq \frac{c_v}{\lambda_v}, \ \forall \ v \in V_c, \\
    &k_v \geq 1, \ \forall \ v \in V_c.
\end{aligned}
\end{align}

Because the solution depends only on the size of the activation sets and the vector $\boldsymbol{k}$, we approximate the solution to~\eqref{eq:relaxedgreedy} in the following way. We start by assuming each node is assigned its own color, and fix each $\omega_v$ to an estimate $\hat{\omega}_v = 1$, for all $v \in V_c$. When these values are fixed, the program is completely independent of $\mathcal{S}$ and becomes a convex program in $\boldsymbol{k}$. We solve this program and greedily color the graph with vertices ordered by weights $k_v$ from the solution. Then, using this new coloring, we update the estimates $\hat{\omega}_v$ with the size of node $v$'s activation set.

We repeat this process iteratively until reaching a vector of estimates $\boldsymbol{\hat{\omega}}$ that no longer improves the solution to~\eqref{eq:relaxedgreedy}. Let $\mathcal{S}_{WGC}$ denote the set of all colorings found while iterating in this way. Then we fix activation sets to those found in the previous iteration, which minimize the solution to~\eqref{eq:relaxedgreedy} over all $\mathcal{S} \in \mathcal{S}_{WGC}$.

Finally, we re-solve $(P2)$ for $\boldsymbol{k_{\mathcal{S}}}$ under this fixed $\mathcal{S}$. Note that this once again enforces integrality of $\boldsymbol{k_{\mathcal{S}}}$, so any solution satisfies the conditions in~\eqref{eq:rateconditions}. Furthermore, note that all vertices that belong to a set $s$ share the same inter-scheduling time $k_s$, so the number of integer variables is reduced from $|V_c|$ in the original problem $(P2)$ to $|\mathcal{S}|$. This can be written as the following mixed-integer convex program,
\begin{align}\label{eq:integersol}
\begin{aligned}
    \min_{\boldsymbol{k_{\mathcal{S}}}} \ &\sum_{s \in \mathcal{S}} \frac{1}{k_s} \\
    \text{s.t.} \ &\sum_{s \in \mathcal{S}} n_{i,s} k_s \leq \tau_i, \ \forall \ f_i \in \mathcal{F}_d, \\
    &k_s \leq \frac{c_v}{\lambda_v}, \ \forall \ v \in s, \ s \in \mathcal{S},\\
    &k_s \in \mathbb{Z}, \ \forall \ s \in \mathcal{S},
\end{aligned}
\end{align}
where $n_{i,s} = |s \cap \mathcal{T}^{(i)}|$ is the number of links in $f_i$'s route that belong to set $s$.

Greedy never uses more than $\Delta(G_c)+1$ colors, so there are at most $\prod_{s=0}^{\Delta(G_c)} k_{s,\max}$ combinations of integer solutions. Here $k_{s,\max}$ is the largest value $k_s$ can take from slice capacity constraints, and is trivially less than the smallest deadline of any flow which traverses it. Therefore, when the degree of the conflict graph is bounded, this can be upper bounded by a constant, and~\eqref{eq:integersol} can be solved with polynomial complexity. In theory, this constant upper bound can be large, but in practice branch-and-bound and cutting plane techniques make modern mixed-integer solvers incredibly efficient at solving integer programs with relatively few variables. The full Weighted Greedy Coloring (WGC) algorithm is shown in Algorithm~\ref{alg:weightedgreedy}.

\begin{algorithm}
    \DontPrintSemicolon
    %
    \SetKwInput{Input}{Input}\SetKwInOut{Output}{Output}
    \Input{Conflict graph $V_c$, rates $\boldsymbol{\lambda}$, deadlines $\boldsymbol{\tau}$, routes $\mathcal{T}$, capacities $\boldsymbol{c}$}
    \Output{Activation sets $\mathcal{S}$ and maximum inter-scheduling times $\boldsymbol{k}_{\mathcal{S}}$}
    Set $\hat{\omega}_v = 1$, for all $v \in V_c$ \\
    Set $\rho^* = \infty$, $\mathcal{S} = []$ \\
    Solve~\eqref{eq:relaxedgreedy} with fixed $\boldsymbol{\omega} = \boldsymbol{\hat{\omega}}$. Denote the solution as $\boldsymbol{k^*}$ \\
    \While{$\rho(\boldsymbol{k^*}) < \rho^*$}{
        Set $\rho^* = \rho(\boldsymbol{k^*})$, $\mathcal{S}^* = \mathcal{S}$ \\
        Initialize activation set $s_1 = \varnothing$, conflict set $\xi_1 = \varnothing$ \\
        Set $\mathcal{S} = [(s_1, \xi_1)]$ \\
        Sort $V_c$ by value of $\boldsymbol{k}^*$, from smallest to largest \\
        \For{each $v \in V_c$}{
            Set $s(v) = 0$ \\
            \For{each $i \in  len(\mathcal{S})$}{
                \If{$v \notin \xi_i$}{
                    Add $v$ to $s_i$ \\
                    Set $s(v) = i$ \\
                    break \\
                }
            }
            \If{$s(v) == 0$}{
                Initialize activation set $s_{i+1} = \{v\}$, conflict set $\xi_{i+1} = \varnothing$ \\
                Set $s(v) = i+1$ \\
                Append $(s_{i+1},\xi_{i+1})$ to $\mathcal{S}$ \\
            }
            \For{$v' \in V_c \setminus v$}{
                \If{$(v,v') \in E_c$}{
                    Add $v'$ to $\xi_{s(v)}$ \\
                }
            }
        }
        \For{each $s \in \mathcal{S}$}{
            \For{each $v \in s$}{
                Set $\hat{\omega}_v = |s|$ \\
            }
        }
        Solve~\eqref{eq:relaxedgreedy} with fixed $\boldsymbol{\omega} = \boldsymbol{\hat{\omega}}$. Denote the solution as $\boldsymbol{k^*}$ \\
    }
    Set $\mathcal{S} = \mathcal{S}^*$ and solve~\eqref{eq:integersol} for $\boldsymbol{k_{\mathcal{S}}}$

    Return $\mathcal{S}$ and $\boldsymbol{k_{\mathcal{S}}}$
    
    \caption{Weighted Greedy Coloring (WGC)}
    \label{alg:weightedgreedy}
\end{algorithm}

\subsection{Schedule Construction}

Given the coloring $\mathcal{S}$ and scheduling vector $\boldsymbol{k_{\mathcal{S}}}$ returned by $WGC$, all that remains is to construct a satisfying schedule using the state-of-the-art pinwheel scheduler $S_{xy}$. Known as a fast online scheduler, $S_{xy}$ takes the vector $\boldsymbol{k_{\mathcal{S}}}$ as input and either produces a schedule with a single activation set $s \in \mathcal{S}$ to be activated in each time slot, or returns false if it cannot find a satisfying schedule. It determines success or failure in $O(|\mathcal{S}|^2)$ time, but returns the schedule on a slot-by-slot basis because the length of the schedule can grow exponentially large, and computing the entire schedule offline may not be possible in polynomial time. 

The algorithm searches for two values $x$ and $y$, and augmented values $k_s' \leq k_s$, for each $k_s$ in $\boldsymbol{k_{\mathcal{S}}}$, such that each $k_s'$ is a power of $2$ multiple of either $x$ or $y$. This results in a partitioning $X$ and $Y$, where we denote the set of values that are multiples of $x$ as $X$, the vector of their inter-scheduling times as $\boldsymbol{k_X}$, and likewise for $Y$ and $\boldsymbol{k_Y}$. By construction, both of these vectors are step-down, and if 
\begin{equation}\label{eq:sxycond}
    \frac{\lceil x \rho(\boldsymbol{k_X}) \rceil}{x} + \frac{\lceil y \rho(\boldsymbol{k_Y}) \rceil}{y} \leq 1,
\end{equation}
then the algorithm returns a satisfying schedule. We omit the full details of the algorithm and encourage readers to reference~\cite{chan_general_1992} for complete details.

Together, $WGC$ and $S_{xy}$ form a polynomial-time algorithm that takes as input a conflict graph $V_c$ with fixed interference constraints and link capacities $\boldsymbol{c}$, fixed rates and deadlines $\boldsymbol{\lambda}$ and $\boldsymbol{\tau}$, and fixed routes $\mathcal{T}$ for all flows. If it is able to find a feasible solution, it returns a deterministic coloring $\mathcal{S}$ of the conflict graph, and a deterministic decision of which activation set to schedule in each time slot. Therefore, running this algorithm in a decentralized fashion at each node generates a coordinated, conflict-free schedule that is guaranteed to support the set of flows $\mathcal{F}$ and requires no control overhead.

\section{Numerical Results}\label{sec:simulations}

In this section we evaluate the performance of our policies in a variety of scenarios, beginning with the solitary flow greedy policy.

\subsection{Solitary Flow Greedy Policy}

Recall that for a solitary flow under total interference and bottleneck slice widths, the greedy policy in Corollary~\ref{cor:greedybound} was shown to be optimal, and we conjectured that it is near-optimal for any value of $\phi$. Furthermore, we showed an approximate bound for the maximum delay seen under a greedy policy based on the idea that the policy will never let a queue grow too much larger than its slice width. Denote the actual maximum packet delay observed after running greedy as $\hat{\tau}(greedy)$, and define the greedy delay ratio
\begin{equation}
    \zeta \triangleq \frac{\lambda \hat{\tau}(greedy)}{\sum_{e \in \mathcal{T}} w_e'}
\end{equation}
as the ratio of this value to the approximate bound on $\tau^*(greedy,\boldsymbol{w'},\lambda)$ in~\eqref{eq:taugreedy}.

\begin{figure}
    \centering
    \includegraphics[width=0.4\textwidth]{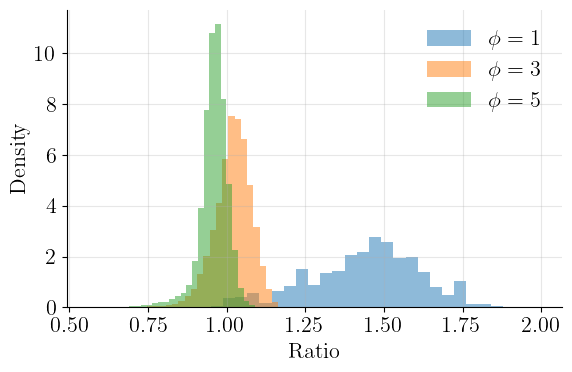}
    \caption{Histogram of solitary flow greedy delay ratios $\zeta$}
    \label{fig:solitary-greedy}
\end{figure}

In Figure~\ref{fig:solitary-greedy}, we show a histogram of $\zeta$ seen after running greedy under a variety of different slice widths and interference models. Using a route with a length of $6$ hops, we varied the value of $\phi$ from $1$ (primary interference) to $5$ (total interference). For each value of $\phi$, we generated $30,000$ random sets of slice widths. For the first third of the sets, we drew values of $\boldsymbol{w}$ from a normal distribution with a mean of $55$ and standard deviation of $15$. For the second third, we drew from a uniform distribution on the interval $[10,100]$, and for the final third we drew from a bimodal distribution, with values normally distributed around $20$ with probability $0.5$ and normally distributed around $100$ otherwise. 

For each set of slice widths, we computed the maximum rate $\lambda^*(\boldsymbol{w})$ and the bottleneck slice widths $\boldsymbol{w'} \leq \boldsymbol{w}$. Figure~\ref{fig:solitary-greedy} shows a histogram of the greedy delay ratio for the combined data with rate $\lambda^*(\boldsymbol{w})$ under these bottleneck slices. In the total interference case, nearly all of the histogram mass is below $1$, indicating that the bound is close to being exact when greedy is optimal. As the value of $\phi$ decreases, we expect more instances to arise where greedy is not optimal as discussed in Section~\ref{sec:isolatedflows}, and we see a corresponding increase in both mean and variance of $\zeta$. Nevertheless, even under primary interference the ratio remains close to $1$. This reinforces greedy as a good heuristic for minimizing deadlines under any value of $\phi$.

\subsection{Weighted Greedy Coloring}

Next we turn to simulating the $WGC + S_{xy}$ algorithm. In an effort to demonstrate the versatility and performance of the algorithm, we show results for a variety of network topologies and traffic patterns. We will also show results comparing it to the Credit-Based Heuristic ($CBH$) algorithm from~\cite{chilukuri_delay-aware_2015}, which is the best known algorithm for meeting deadlines in multi-hop wireless networks with interference and fixed arrivals. Our results show that we are able to significantly improve on the $CBH$ algorithm and meet deadlines on the order of milliseconds over multiple hops, while remaining close to throughput optimality.

\begin{figure}
    \centering
    \includegraphics[width=0.4\textwidth]{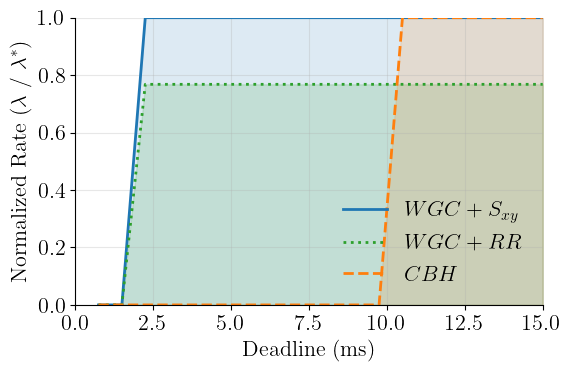}
    \includegraphics[width=0.4\textwidth]{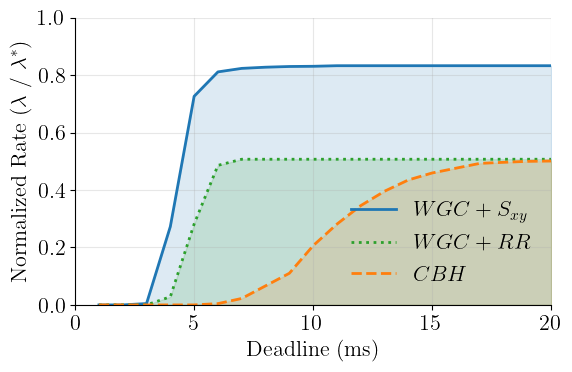}
    \includegraphics[width=0.4\textwidth]{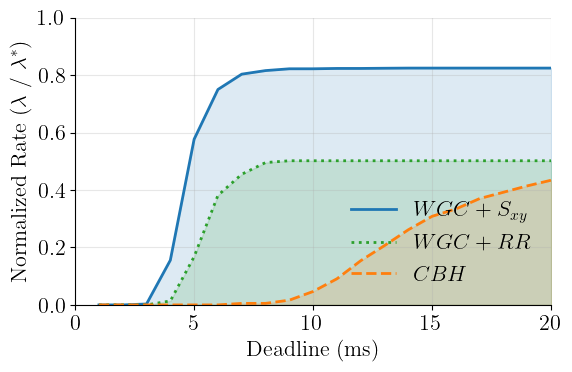}
    \caption{Average feasibility regions for a sink tree topology with depth $3$ and degree $4$ (top), a $4x4$ mesh grid (middle), and random topologies generated by randomly failing links in a $4x4$ mesh grid with probability $p=0.25$ (bottom).}
    \label{fig:lambda-tau-feas}
\end{figure}

For ease of exposition, in each experiment we set rates and deadlines equal to the same $\lambda$ and $\tau$ respectively across all flows, and we generate $32$ flows with random source/destination pairs unless otherwise specified. We set the duration of each time slot to be $125 \mu s$, which is the smallest time slot typically used by data in the 5G standard~\cite{3gpp-ts-138211}. The $CBH$ algorithm assumes that all rates are integer-valued, so for the most accurate comparison we round each rate to the nearest integer when comparing our results to $CBH$. In each plot, we show the average of $200$ random instances, where the randomness is over the flow generation and topology (when random). This is a sufficiently large sample so that the values shown are within $5\%$ of the true value with $95\%$ confidence.

Under primary interference ($\phi=1$), we are able to compute the throughput-optimal $\lambda^*$ independent of deadlines in polynomial time, using the well-known algorithm from Hajek and Sasaki~\cite{hajek_link_1988}. In many of our experiments, we normalize the achievable rates under each algorithm by $\lambda^*$, so that if an algorithm achieves a normalized rate of $1$ for a given deadline $\tau$, it means that this deadline can be achieved without any loss in throughput. The $(\lambda,\tau)$ feasibility curve shows the loss in throughput as deadlines become more strict. Note that even as deadlines become large, our algorithm may not be able to achieve throughput optimality because the $WGC$ algorithm forces each link to belong to a single activation set.

Recall that our algorithm takes as input a set of flows with known rates and deadlines, and attempts to find a feasible schedule to simultaneously satisfy all service guarantees. If an algorithm is able to find such a schedule for a given $(\lambda,\tau)$ pair, then that pair lies within the feasibility region for that algorithm. In Figure~\ref{fig:lambda-tau-feas}, we show the $(\lambda, \tau)$ feasibility region under each algorithm and for several different topologies and traffic patterns with $\phi = 1$.

The top of the figure shows results for a sink-tree topology with depth $3$, degree $4$, and a single flow arriving at each leaf. This makes the traffic highly symmetric and amenable to the $CBH$ algorithm, which performs best under tree topologies. As a result, $CBH$ is able to achieve throughput optimality for all deadlines above $\approx 10 ms$. $WGC+S_{xy}$ is able to do the same, however, for deadlines above $\approx 2.5 ms$. This shows that our algorithm is able to take similar advantage of the symmetry and guarantee throughput optimality with less than $1 ms$ of latency per hop. $WGC$ combined with a simple round-robin schedule is able to achieve the same deadlines but only about $75\%$ of the optimal throughput.

In the middle figure, we show results for a $4x4$ mesh grid topology and $32$ flows with random source/destination pairs. $WGC + S_{xy}$ can achieve significantly tighter deadlines and higher throughput in this setting compared to $CBH$. The bottom figure shows similar results for a random topology, where we start with a $4x4$ mesh grid and allow links to fail independently with probability $p=0.25$ (repeating this process if the resulting graph is not connected). As topologies and traffic flow become less symmetric, $WGC+S_{xy}$ is able to handle this gracefully and the performance gap increases significantly between $WGC+S_{xy}$ and the other policies.

\begin{figure}
    \centering
    \includegraphics[width=0.4\textwidth]{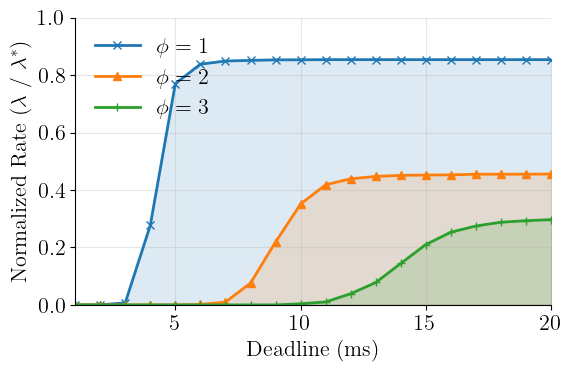}
    \includegraphics[width=0.4\textwidth]{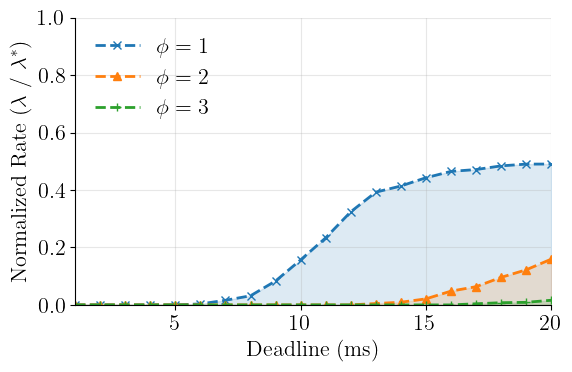}
    \caption{Average feasibility regions for a $4x4$ mesh grid topology and varying values of $\phi$, showing the performance of $WGC+S_{xy}$ (top) and $CBH$ (bottom). All rates are normalized with respect to the optimal $\lambda^*$ when $\phi=1$.}
    \label{fig:phi-sweep}
\end{figure}

We next examine how the interference model affects performance. In Figure~\ref{fig:phi-sweep}, we show the feasibility regions for a $4x4$ mesh grid under $\phi$-hop interference for varying values of $\phi$. All rates are normalized with respect to the value of $\lambda^*$ under $\phi=1$, so the comparison shown in each plot is the relative achievable rate across values of $\phi$. Naturally, as $\phi$ increases, both achievable rates and deadlines suffer as fewer links are able to be activated in each slot. Once again, the performance of $WGC+S_{xy}$ shows significant improvement over $CBH$ under any value of $\phi$.

\begin{figure}
    \centering
    \includegraphics[width=0.45\textwidth]{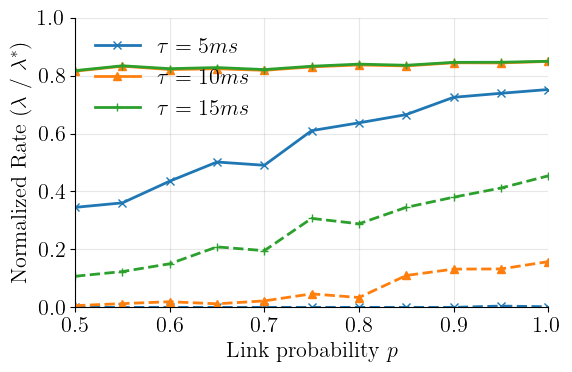}
    \caption{Average feasible normalized rate for a $4x4$ mesh grid with $\phi=1$, where each link is on iid with probability $p$, and with varying deadlines. The solid curves show the performance of $WGC+S_{xy}$ and the dashed curves show the performance of $CBH$.}
    \label{fig:link-prob-sweep}
\end{figure}

We observed in Figure~\ref{fig:lambda-tau-feas} that when links randomly fail in a $4x4$ mesh grid, the performance of each algorithm suffers. Figure~\ref{fig:link-prob-sweep} shows the change in performance as topologies become less structured and more random, with solid curves representing $WGC+S_{xy}$ and dashed curves representing $CBH$. We start with a $4x4$ mesh grid, and determine whether each link is active iid with some probability $p$. We again plot the normalized rate for each algorithm and three different deadline values as we sweep $p$ from $0.5$ to $1$. When $p=1$, we recover the $4x4$ mesh grid, and as $p$ becomes smaller, the topology becomes less structured. In any instance, if the resulting topology is not connected, we regenerate the topology. The figure shows that each algorithm improves with more structure in the topology, but for deadlines above $10 ms$, the change in performance for $WGC+S_{xy}$ is negligible, showing its versatility under any network topology.

\begin{figure}
    \centering
    \includegraphics[width=0.45\textwidth]{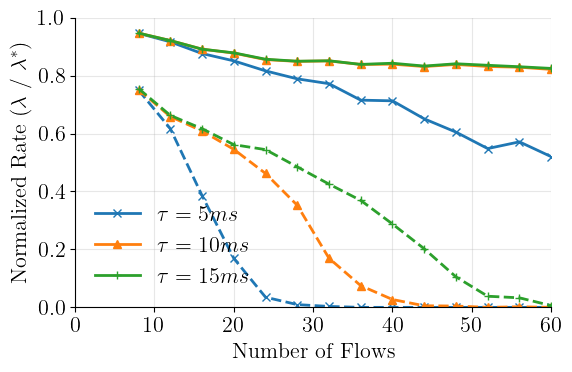}
    \caption{Average feasible normalized rate for a $4x4$ mesh grid with $\phi=1$, and with varying number of flows and deadlines. The rates are normalized separately for each set of flows, so the performance dropoff is not due to network saturation but rather scheduling limitations. The solid curves show the performance of $WGC+S_{xy}$ and the dashed curves show the performance of $CBH$.}
    \label{fig:num-flows-sweep}
\end{figure}

Finally, we examine how the number of flows and different service guarantees that must be met simultaneously affect performance. In Figure~\ref{fig:num-flows-sweep}, we plot the normalized rate as a function of the number of flows, again denoting $WGC+S_{xy}$ with solid curves and $CBH$ with dashed curves. We show results for a $4x4$ mesh grid topology and normalize rates separately for each instance, so that any performance dropoff as the number of flows increases is not due to network saturation, but rather the ability to schedule more flows simultaneously. The performance of each algorithm drops off as the number of flows increases, as it becomes harder to satisfy deadlines for more unique source/destination pairs, but once again the dropoff is small under $WGC+S_{xy}$ for deadlines above $10 ms$.

Each of these results shows the ability of $WGC+S_{xy}$ to meet tight deadlines on the order of milliseconds across multiple hops in wireless networks with general topology, interference, and traffic patterns. Moreover, our algorithm is able to meet these deadlines while achieving near-optimal throughput in many cases. It accomplishes this in polynomial time and determines offline whether a network can support a given set of flows, making it a significant step forward in scheduling to meet tight deadlines in wireless networks with interference.

\section{Conclusion}\label{sec:conclusion}

In this paper, we analyzed the impact of wireless interference on scheduling for service guarantees. We defined throughput- and deadline-optimal policies for a solitary flow, and showed that even when queues are stable, packets can experience large delay. To alleviate this problem, we derived conditions on end-to-end delays in terms of inter-scheduling times, and showed that we can meet tight deadline guarantees under any interference model by solving a generalized version of pinwheel scheduling. Finally, we developed a heuristic algorithm to solve this problem in polynomial time, which can achieve deadlines on the order of milliseconds and near-optimal throughput under arbitrary network topologies and traffic patterns. Future work includes optimizing routing as well as scheduling, and expanding our policy to support unreliable links, changing wireless topologies, and stochastic traffic.

\appendix

\subsection{Proof of Theorem~\ref{th:cyclicschedules}}\label{app:cyclicschproof}

We start by showing that a policy $\pi'$ supports $\mathcal{F}$ if and only if at most $\lambda_i \tau_i$ packets from $f_i$ are present in the system at the end of each slot. To show the forward direction, recall that arrival rates are fixed on the interval $0 \leq t \leq T$, so exactly $\lambda_i$ packets of $f_i$ arrive in each slot. Packets are served in a FCFS manner, so all packets must have arrived in the last $\tau_i$ slots. Similarly, if there were more than $\lambda_i \tau_i$ packets, then at least one must have arrived prior to the last $\tau_i$ slots, so the condition is both necessary and sufficient.

Define the state of the system at time $t$ under a supporting policy $\pi'$ as the length of each queue and denote it by $\boldsymbol{Q^{\pi'}}(t)$. Queue lengths are bounded by the argument above, and there exists a finite number of states which can be visited over any time horizon $T$. In particular, for sufficiently large $T$, there must exist a state $\boldsymbol{Q^{\pi'}}(t_0)$ which occurs at time $t_0$ and then occurs again at time $t_0+K$ for some $K > 0$. Let $\boldsymbol{Q^{\pi'}}(t_0)$ be the first state where this event occurs, and let the set of actions taken in the time interval $[t_0,t_0+K)$ be the policy $\pi$ with $K^{\pi} = K$ and $\mu^{\pi}(t) = \mu^{\pi'} \big((t-t_0) \bmod K^{\pi} + t_0 \big)$ for all $t \geq 0$. Then for all $t_0 \leq t \leq T$, we have $\boldsymbol{Q^{\pi}}(t) = \boldsymbol{Q^{\pi'}} \big((t-t_0) \bmod K^{\pi} + t_0 \big)$, so~\eqref{eq:queuebound} holds and $\pi$ supports $\mathcal{F}$ for all $t_0 \leq t \leq T$.

Note that $t_0 > 0$ because queues are empty at $t=0$ and the system requires time to ``ramp up''. However, we next show that $\pi$ supports $\mathcal{F}$ on the interval $0 \leq t < t_0$ using induction on the queue sizes. The queue evolution equations are given by 
\begin{equation}
    Q_{i,e}^{\pi}(t+1) = \min\{ Q_{i,e}^{\pi}(t) + \tilde{\lambda}_{i,e}(t) - \mu_e^{\pi}(t) w_{i,e}, \ 0 \},
\end{equation}

for all $0 \leq t \leq T$, where
\begin{equation}
    \tilde{\lambda}_{i,e}(t) = \begin{cases}
        \lambda_i, &e = \mathcal{T}^{(i)}_0, \\
        \mu_{e_{-1}}^{\pi}(t) \cdot \min \{w_{i,e_{-1}}, \ Q_{i,e_{-1}}^{\pi}(t) \}, &e \neq \mathcal{T}^{(i)}_0,
    \end{cases}
\end{equation}

and with a slight abuse of notation we denote the link preceding $e$ in $\mathcal{T}^{(i)}$ as $e_{-1}$, where it is understood we are referring to $f_i$. Now assume that $Q_{i,e}^{\pi}(t) \leq Q_{i,e}^{\pi}(t+K^{\pi})$ for all $i$ and $e$. Then, from the queue evolution equations and given that $\mu^{\pi}(t) = \mu^{\pi}(t+K^{\pi})$, this implies that $Q_{i,e}^{\pi}(t+1) \leq Q_{i,e}^{\pi}(t+K^{\pi}+1)$ for all $i$ and $e$ as well. By definition, $Q_{i,e}^{\pi}(0) \leq Q_{i,e}^{\pi}(K^{\pi})$ for all $i$ and $e$, which completes the induction step. Therefore,
\begin{equation}
    \sum_{e \in \mathcal{T}^{(i)}} Q_{i,e}^{\pi}(t) \leq \sum_{e \in \mathcal{T}^{(i)}} Q_{i,e}^{\pi}(t+nK^{\pi}) \leq \lambda_i \tau_i
\end{equation}

for all $f_i$ and $0 \leq t < t_0$, and some $n>0$ such that $t_0 \leq t+nK^{\pi} < t_0 + K^{\pi}$, which shows that $\pi$ supports $\mathcal{F}$ on the interval $0 \leq~t < t_0$.

It can easily be shown that $\pi$ also supports $\mathcal{F}$ on the interval $t > T$. Assume for all $t > T$, dummy packets arrive with the same rate $\lambda_i$ for each flow, until all real packets have been delivered. Then, because $\pi$ supports every flow under regular arrivals, it must also support every flow under the dummy arrivals. This completes the proof.

\subsection{Proof of Corollary~\ref{cor:ratestability}}\label{app:stabilityproof}

We show this by contradiction. Assume $\bar{\mu}_e^{\pi} < \lambda_i \ w_{i,e}$ for some $f_i$ and $e \in T^{(i)}$, so equivalently
\begin{equation}\label{eq:stabilityeq}
    \bar{\mu}_e^{\pi} K^{\pi} w_{i,e} = \eta_e^{\pi} w_{i,e} < \lambda_i K^{\pi}.
\end{equation}
Recall that $\eta_e^{\pi}$ is the number of times a link is scheduled in a scheduling period, and each time it serves $\min\{ w_{i,e}, Q_{i,e}(t)\}$ packets. Then the total number of packets served by link $e$ in any period is at most the left-hand side of~\eqref{eq:stabilityeq}, while the right-hand side is the total number of packets that arrive to the network in that period.

Therefore, after one period, the total number of $f_i$ packets in flight from the source to link $e$ is at least the difference in these quantities. Because the inequality is strict, this must be at least as large as some $\epsilon > 0$. This holds for each scheduling period, so after $\Omega$ scheduling periods, the total number of $f_i$ packets in flight is at least $\epsilon \Omega$. Now let $\Omega \to \infty$, so this quantity also goes to infinity because $\epsilon$ is strictly positive.

In order for a policy to support $\mathcal{F}$, all queue sizes must be bounded by a finite constant, else the policy cannot meet any finite deadline. There are a finite number of queues along $T^{(i)}$, so at least one of them must also grow without bound, which is a contradiction.

\subsection{Proof of Theorem~\ref{th:orrbound}}\label{app:orrproof}

The $ORR$ policy schedules links in order from source to destination in subsequent time slots according to the definition above. The activation rate follows because each link is only activated once per scheduling period. Now assume a packet arrives at the source at some time $t$, and that slice widths are large enough to serve all enqueued packets when a link is scheduled. This is non-restrictive because $\tau_i^*$ is defined as the smallest feasible deadline for an arbitrarily small throughput. At time $t \bmod{(\phi+1)}$, the packet is served at link $\mathcal{T}^{(i)}_0$, and following the $ORR$ schedule, it is served at each subsequent link in the next $|\mathcal{T}^{(i)}|-1$ slots. In the worst case, the packet must wait $\phi$ slots at the source before being served, so it spends a maximum of $|\mathcal{T}^{(i)}|+\phi$ slots in the network, which verifies $\tau_i^*$ for the $ORR$ policy.

It remains to show that no other policy can achieve a smaller deadline for all packets. Recall that only one of $\{\mathcal{T}^{(i)}_0, \dots, \mathcal{T}^{(i)}_{\phi} \}$ can be scheduled in the same slot. We claim that it must take at least some packets $2\phi+1$ slots to reach link $\mathcal{T}^{(i)}_{\phi+1}$. The fewest slots a packet can take is $\phi+1$, so we define $\Delta(t)$ as the number of additional slots it takes beyond this minimum for packets which arrive at the source at time $t$. Then if $\Delta(t) \geq \phi$ for any $t$, our claim must hold. Note that $\Delta(t)=0$ at time $t$ when packets arrive, and it is incremented by one each time a scheduling decision is made that does not schedule those packets.

Assume our claim does not hold, i.e., that $\Delta(t) < \phi$ for all $t$. First note that packets which arrive at $t$ can allow at most $\Delta(t)$ slots of arrivals behind them to ``catch up''. We say a slot of arrivals $t+j$ is caught up to $t$ if it is enqueued at the same link as the slot of arrivals $t$. Each slot of arrivals which catches up to $t$ increments $\Delta(t)$, so arrivals from slot $t+\phi$ cannot be caught up to $t$ under our assumption on $\Delta(t)$. Let $t+j^* \leq t+\phi$ be the first slot which is not caught up to $t$ when packets from slot $t$ reach $\mathcal{T}^{(i)}_{\phi+1}$. Because the previous slot of arrivals is caught up to $t$, it must have been scheduled $\phi$ times independently of the arrivals from slot $t+j^*$. Therefore, $\Delta(t+j^*) \geq \phi$, which is a contradiction. This proves the result.

\subsection{Proof of Corollary~\ref{cor:greedybound}}\label{app:greedyboundproof}

In any policy with bottleneck slice widths, a link $e$ must serve a full slice width of packets each time it is scheduled (after the ramp-up period). If not, it would serve fewer than $\eta_e w_e = \lambda K^{\pi}$ packets per scheduling period, and the queue would become unstable. Therefore the first condition of greedy must always hold for at least one link $e$ after the ramp-up period.

Now assume there are two queues which are ``full,'' i.e., contain at least a slice width of packets at time $t$, and denote them as links $e$ and $e'$ respectively. Without loss of generality let link $e'$ be the closest full queue to the destination, and assume that link $e$ is scheduled at $t$. We will show that this cannot be optimal by contradiction. When link $e$ is served, $w_e$ packets are served to the next queue in the route, but remain behind the packets already queued at link $e'$. Future actions are taken, and at some point $t' > t$, link $e'$ is scheduled. No queues could have changed between link $e'$ and the destination on the interval $[t,t')$, because only full queues can be served. Furthermore, because $Q_{e'}$ was already full at time $t$ and queues are FCFS, only the packets which would have been served at $t$ are now served at $t'$.

If link $e'$ is not the destination link, then the state of the system after $t'$ is exactly the same as if link $e'$ had been scheduled at time $t$, and the actions in the interval $[t,t')$ were all delayed by one slot. If link $e'$ is the destination link, then the $w_{e'}$ packets served by link $e'$ could have arrived at their destination sooner had link $e'$ been scheduled at $t$. Therefore, it could not have been optimal to serve link $e$, which is a contradiction.

Bounding the exact value of $\tau^*(greedy, \boldsymbol{w'}, \lambda)$ is challenging because the scheduling decisions at each step of greedy depend on the state of the queues, which creates a recursive relationship. We can approximate $\tau^*$ by observing that $Q_e$ is never allowed to grow much larger than $w'_e$ before being served, and that if link $e$ is scheduled, the subsequent queue $Q_{e+1}$ is strictly smaller than $w'_{e+1}$. Therefore $\sum_{e \in \mathcal{T}} Q_e(t) \lesssim \sum_{e \in \mathcal{T}} w'_e$ at any time $t$, and the result follows from Corollary~\ref{cor:deadlinecondition}.

\subsection{Proof of Theorem~\ref{th:deadlinebound}}\label{app:deadlineboundproof}

Assume without loss of generality that slices have bottleneck slice widths. Allowing slice widths to be larger than this can only decrease packet delay, so this assumption is not only non-restrictive but necessary. Under bottleneck slice widths, $\bar{\mu}_{i,e}^{\pi} w_{i,e} = \lambda_i$ for all $f_i \in \mathcal{F}$ and $e \in \mathcal{T}^{(i)}$ by definition, so the total number of $f_i$ packets served in each period in steady state is $\eta_e^{\pi} w_{i,e} = \bar{\mu}_{i,e}^{\pi} w_{i,e} K^{\pi} = \lambda_i K^{\pi}$. Because this many packets are served at each link, this is the number of arrivals each subsequent link sees per scheduling period, as well as the number of arrivals at the source link.
    
Assume that $Q_{i,e}(t) = 0$ at some time $t$. Then because it serves the same number of packets as arrivals within each scheduling period, the queue must be emptied again at some time $t_{i,e}(t) \leq t + K^{\pi}$. By induction, the queue must be emptied at least once per scheduling period if it is ever empty. Assume that before the network reaches steady-state, dummy packets are added to each queue so that links still serve a full slice width. Because queues are initialized to zero, this completes the induction step.

Let $t_{i,e}$ be any time slot when $Q_{i,e}(t_{i,e}) = 0$. Because link $e$ is scheduled $\eta_e$ times per scheduling period and serves a full $w_{i,e}$ packets each time, the maximum delay that packets can experience at link $e$ occurs when link $e$ is scheduled the $\eta_e$ time slots leading up to time $t_{i,e}$, regardless of when packets arrived. 

Now consider the next link $e+1$ in the route, and let $t_{i,e+1}$ be the first time slot after $t_{i,e}$ where $Q_{i,e+1}(t_{i,e+1}) = 0$. Just like link $e$, the maximum delay packets can experience occurs when $e+1$ is scheduled in $\eta_{e+1}$ consecutive slots leading up to time $t_{i,e+1}$, and where the $\lambda_i K^{\pi}$ packets from link $e$ arrive in the interval $t_{i,e+1} - K^{\pi} < t \leq t_{i,e+1}$. In particular, the worst-case delay occurs when link $e+1$ begins receiving packets from link $e$ at time $t_{i,e+1}-K^{\pi}+1$, immediately after $Q_{i,e+1}$ becomes empty. The same holds for each link in $\mathcal{T}^{(i)}$, and so the worst-case delay occurs when links are scheduled in blocks such that each link $e$ begins scheduling its block immediately after link $e+1$ finishes scheduling.

This is a valid schedule with bottleneck links under any $\phi$. To see this, consider any set of $\phi+1$ consecutive links $\{e,\dots,e+\phi\}$. From interference constraints, $\sum_{j=0}^{\phi} \bar{\mu}_{e+j} \leq 1$, and each link $e'$ is part of at least one set of $\phi+1$ consecutive links where this bound is tight. If not, $\bar{\mu}_{e'}$ could be increased and $w_{i,e'}$ could be decreased, which contradicts our bottleneck slice assumption. Then for each set of $\phi+1$ links, scheduling blocks in reverse order of packet flow fills the entire scheduling period and satisfies activation rates. When two such sets of links overlap or adjoin, the schedule for each block is shifted so the pattern holds for all sets of $\phi+1$ links.

Because packets are scheduled in blocks, we consider the first and last packet in each block. Denote the time the first packet arrives at the source as $t(s_0)$ and the time the last packet in the block arrives at the source as $t(s_{-1}) = t(s_0) + K^{\pi}-1$. Similarly, denote the time the first packet is delivered to its destination as $t(d_0)$ and the time the last packet is delivered as $t(d_{-1}) = t(d_0) + \eta_{-1} -1 < t(d_0) + K^{\pi}-1$, where link $-1$ is the destination link. The total delay seen by the first packet is therefore larger than the last packet, and because all intermediate packets see a gradient of delay between these two values, the first packet sees the largest delay. This packet is the first served in each block, so at each link $e$ it sees a delay of $K^{\pi}-\eta_e^{\pi} = K^{\pi}(1 - \bar{\mu}_e^{\pi})$ before being served. Summing over the route yields the result.

This bound is tight for at least one policy by construction of the policy just described. We say that there are many policies where $\tau_i^*$ grows linearly with $K^{\pi}$, because any policy which schedules links in blocks incurs a delay at each link that is linear in $K^{\pi}$. In fact, from Lemma~\ref{lemma:taulowerbound}, any policy with inter-scheduling times that depend on the schedule length incurs such a delay.

\subsection{Proof of Lemma~\ref{lemma:delaydeficit}}\label{app:delaydeficit}

We will show this by induction, by first showing that if packets arrive at link $e$ with delay deficit less than or equal to zero, then whenever they arrive at the next link in the route, their delay deficit will also be at most zero. Let $Q_{i,e}$ have a counter which tracks the number of packets with strictly positive delay deficit. Because $\delta^l(t) \leq 0$ for all packets $l$ by assumption, no packets are added to the counter on arrival. Of the packets already in the queue, \eqref{eq:delaydeficit} shows that a strictly positive delay deficit implies that a packet's age $a^l(t) > \sum_{\tilde{T}^{(l)}} k_e^{\pi}$ at time $t$. Because $\lambda_i$ packets arrive at the source in each time slot, at most $\lambda_i$ new packets can exceed this bound in each slot, so at most $\lambda_i$ packets are added to the counter in each slot.

Next we claim that when link $e$ is scheduled, it serves all packets with positive delay deficits. Assume this is true at some slot $t$ when link $e$ is scheduled. Then due to the maximum inter-scheduling time, it must be scheduled again within the next $k_e^{\pi}$ slots. We just saw that at most $\lambda_i$ packets can be added to the counter in each slot, so at most $\lambda_i k_e^{\pi}$ packets can have positive delay deficit the next time link $e$ is scheduled. This is a lower bound on the allowable slice width, so all these packets are again served. The first time that link $e$ is scheduled, there can be no more than $\lambda_i k_e^{\pi}$ packets in the network, so it necessarily serves all packets with positive delay deficits, and this completes the induction step.

Therefore, every packet enqueued at link $e$ is served no later than the first time that link $e$ is scheduled after its delay deficit becomes positive. Because the link is scheduled at least every $k_e^{\pi}$ slots, the delay deficit of any packet is at most $k_e^{\pi}$ when it is served. This is exactly the delay quota of link $e$, so no packet can have a delay deficit larger than zero when it arrives at the next link in its route. By definition, packets arrive at the source link with a delay deficit of zero, which completes the induction step and the proof.

\bibliographystyle{IEEEtran}
\bibliography{ton_submission}


\end{document}